\documentclass[12pt, oneside, reqno]{amsart}
\pdfoutput=1
\usepackage{a4wide}
\usepackage{amssymb}
\usepackage{amsthm}
\usepackage{amsmath}
\usepackage{amscd}
\usepackage[foot]{amsaddr}
\usepackage{verbatim}
\usepackage{color}
\usepackage[mathscr]{eucal}
\usepackage{commath}
\usepackage{tikz}

\usepackage[font=small]{caption}
                     
\usepackage[margin=0.8in]{geometry}
\usepackage{latexsym,amssymb,amstext,amsmath}
\usepackage{amsfonts}

\numberwithin{equation}{section}

\usepackage{graphicx}

\usepackage{color}
\usepackage{comment}

\usepackage{amsthm}

\usepackage{babel}
\usepackage{amsmath}
\usepackage{amsfonts}
\usepackage{enumerate}
\usepackage{geometry}
\usepackage{url}
\usepackage{parskip}
\usepackage{amssymb}
\usepackage{graphicx}
\usepackage{musicography}
\usepackage{color}
\usepackage{comment}
\usepackage{amsthm}
\usepackage{cite}

\makeatletter
\renewcommand\subsection{\@startsection{subsection}{2}%
  \z@{.5\linespacing\@plus.7\linespacing}{.5\linespacing}%
  {\normalfont\scshape}}
\renewcommand\subsubsection{\@startsection{subsubsection}{3}%
  \z@{.5\linespacing\@plus.7\linespacing}{-.5em}%
  {\normalfont\scshape}}
\makeatother

\usepackage{float}
\usepackage[caption = false]{subfig}

\usepackage{hyperref}

\newtheorem{theorem}{Theorem}
\newtheorem{corollary}{Corollary}[theorem]
\newtheorem{lemma}{Lemma}
\newtheorem{proposition}{Proposition}

\newgeometry{vmargin={30mm}, hmargin={30mm,30mm}} 

\def\be {\begin{equation}}
\def\ee {\end{equation}}
\def\bea {\begin{eqnarray}}
\def\eea {\end{eqnarray}}

\renewcommand{\=}{\; = \;}

\date{}

\begin{document}
	
\vspace{5cm}	
\title{The immortal dyon index is positive}
\author{Mart\'i Rossell\'o}
\email{marti@gate.sinica.edu.tw}

\maketitle

\bigskip

\maketitle
\thispagestyle{empty}
\centerline{ \it ~Institute of Mathematics, Academia Sinica, Taipei, Taiwan.}

\bigskip

\begin{abstract}  
The microstates of supersymmetric black holes in asymptotically flat four-dimensional spacetime are expected to be bosonic due to the spherical symmetry of their horizons. This implies that the index counting the difference between bosonic and fermionic black hole microstates must be positive, as conjectured by Sen in \cite{Sen:2011ktd}. We show that the conjecture is satisfied by the index which counts $\frac{1}{4}$-BPS states in $\mathcal{N}=4$ string theory forming single centered black holes,  called immortal dyons. The proof relies on the relation between the indexed degeneracies and the coefficients of a family of mock modular forms%
, thus showing that black holes predict the sign of the Fourier coefficients of mock modular forms. %

\end{abstract}

\newpage

\tableofcontents

\newpage

\section{Introduction and statement of results}

Black holes have been a major source of insight into quantum gravity. The holographic principle and its concise realisation in the AdS/CFT correspondence can ultimately be seen as a consequence of the Bekenstein-Hawking area law for black hole entropy.

In this work, black holes are the source of a non-trivial mathematical statement. Namely, we show that black holes predict the sign of the Fourier coefficients of a meromorphic Siegel modular form, as conjectured by Sen in \cite{Sen:2011ktd}, or equivalently, the sign of the Fourier coefficients of a family of mock modular forms.%

 Such a precise statement originates from the study of supersymmetric black holes in string theory, for which the microscopic origin of the Bekenstein-Hawking entropy has long been understood as coming from the large degeneracy of bound states of strings and branes \cite{Sen:1995in, Strominger:1996sh}. The quantum generalization of the Bekenstein-Hawking-Wald entropy for extremal black holes, called the quantum entropy function, can be  defined as a path integral over string fields by making use of the AdS$_2$ factor in their near-horizon geometry \cite{Sen:2008vm}.
 This definition of quantum entropy allows for reproducing the counting of supersymmetric black hole microstates from the gravitational path integral by using supersymmetric localization in supergravity \cite{Banerjee:2009af,Dabholkar:2010uh,Dabholkar:2011ec,Gupta:2012cy,Dabholkar:2014ema,Gupta:2015gga,Murthy:2015yfa,Murthy:2015zzy,Jeon:2018kec,deWit:2018dix,Iliesiu:2022kny,LopesCardoso:2022hvc,Sen:2023dps,GonzalezLezcano:2023uar}.  The %
growing body of evidence points towards
the exact equivalence between the two complementary pictures of a black hole given by string theory: %
a bound state of strings and branes, %
and a solution to the supergravity equations of motion with an event horizon. %

The enumeration of states (or multiplets) at weak string coupling is performed through a supersymmetric index, which counts the difference between bosonic and fermionic BPS states annihilated by a fraction of the supercharges. The agreement between the number of indexed microstates in the limit of large charges and the semiclassical answer for the entropy indicates that at strong coupling, and in the large charge limit, the number of bosonic states must be much larger than the number of fermionic states.

For supersymmetric black holes in 4d asymptotically flat spacetime, a stronger argument can be made by exploiting the spherical symmetry of their horizons,\footnote{The heuristic argument for extremal black holes in 4d asymptotically flat spacetime is based on their solitonic behaviour. Being solutions to the equations of motion, the curvature of the AdS$_2$ factor in their near-horizon geometry can only be cancelled by an $S^2$ factor. 
The precise argument for the spherical symmetry of black holes preserving 4 supersymmetries follows by demanding the closure of the supersymmetry algebra \cite{Dabholkar:2010rm}, which implies that the symmetry group contains an $SU(2)$ factor which can be identified with a subgroup of spatial rotations.}
which implies that they carry zero angular momentum. Furthermore, the quantization of the fermion zero modes associated to the breaking of some supersymmetries by the black hole does not affect the spherical symmetry of the horizon because these zero modes live outside the horizon \cite{Banerjee:2009uk, Jatkar:2009yd}. Since extremal black holes possess an AdS$_2$ factor in their near horizon geometry, using AdS$_2$/CFT$_1$ arguments \cite{Sen:2008vm}, Sen postulated that supersymmetric black holes in four dimensions form an ensemble of states all of which carry zero angular momentum \cite{Sen:2009vz}, therefore providing a sharp prediction for the sign of the index for all charges \cite{Sen:2011ktd}.\footnote{Naively, one would expect the ensemble of black hole microstates to have zero average angular momentum, but the AdS$_2$ boundary conditions fix the ensemble to be the microcanonical one where the charges are fixed.}

The exact matching between the index counting states at weak string coupling and the index defined by the quantum entropy function, computed by supersymmetric localization of the path integral in supergravity, was obtained in \cite{Iliesiu:2022kny}. These results give additional support to the previous positivity conjecture by arguing that the change of boundary conditions from the gravitational index to the path integral computing the black hole microstate degeneracies leaves the answer invariant. 

The preceding discussion can be summarized in the following schematic form,
\begin{equation}
    \text{Tr}\,(-1)^{F} \= \text{Tr} \,1 \,,
\end{equation}
where the trace is taken over the black hole microstates with fixed charges and $F$ is the fermion number, which can be related to the third component of angular momentum $J$  through $(-1)^F = e^{2\pi i J}$.

The previous arguments apply to the index counting horizon degrees of freedom of a black hole. However,  the counting of microstates is typically done through a spacetime helicity supertrace \cite{Bachas:1996bp, Gregori:1997hi}. In the present work, we are interested in the index counting $\frac{1}{4}$-BPS states in unorbifolded $\mathcal{N}=4$ string theory compactified to four dimensions, meaning heterotic string theory compactified on $T^6$, or its dual type II compactified on $K3\times T^2$, for which a formula for the indexed degeneracies was first conjectured in \cite{Dijkgraaf:1996it}. For these states, which break $12$ of the $16$~supersymmetries, the pertinent index is the sixth helicity supertrace, which we introduce in the next section. The relevant point is that the trace computes the index of asymptotic states, and not just horizon states, therefore possibly receiving additional contributions external to the horizon. The additional contributions to the helicity supertrace can be of three types \cite{Sen:2011ktd} --  hair modes, multi-centered black holes and horizonless smooth solutions. We expand on the treatment for the first two cases and refer the reader to \cite{Sen:2011ktd} for the reasons why no contributions from horizonless solutions to the index are expected.

Hair modes are normalizable fluctuations which have support outside the horizon of the black hole. The quantization of the fermionic hair modes which are not the fermionic zero modes associated with broken supersymmetry can produce states which preserve supersymmetry and contribute to the index, therefore possibly spoiling the positivity argument. Since the hair modes studied in \cite{Banerjee:2009uk} arise from momentum along the compact directions, in a frame where all the charges of the black hole are pure D-brane charges there are no expected extra hair modes (as in the case of $\frac{1}{8}$-BPS states in~$4d $~$\mathcal{N}=8$ string theory). We assume that such frame exists for $\frac{1}{4}$-BPS black holes. For these black holes, hair modes were studied in the type IIB frame \cite{Banerjee:2009af, Chattopadhyaya:2020yzl} and shown to be bosonic thereby not spoiling the positivity conjecture.\footnote{In fact, in this frame they provide a more stringent conjecture than the one we prove in this work.}

The previous arguments, which dealt with the horizon of a single black hole, break down for multi-centered black holes since in this case the whole configuration can have angular momentum. Since the helicity trace index  at a generic point in the moduli space and for generic charges can receive contributions from multi-centered black holes, one must either subtract their contribution to the index or work at a point in the moduli space where there are no contributions from multi-centered black holes. For the index counting $\frac{1}{4}$-BPS states in $\mathcal{N}=4$ string theory, the multi-centered contributions come necessarily from bound states of two $\frac{1}{2}$-BPS states \cite{Dabholkar:2009dq}. These contributions are under control as there is a prescribed way to extract the single centered contributions from the index \cite{Cheng:2007ch}. Single centered black holes, which carry both electric and magnetic charges, exist at every point in the moduli space and are therefore called immortal dyons. We call the index counting their degeneracies the immortal dyon index.

Postponing the precise expression of the index to the next section, it now suffices to know that the charges carried by the $\frac{1}{4}$-BPS states can be encoded in three integers $(m,n,\ell)$ and we denote immortal dyon index, the indexed degeneracy for $\frac{1}{4}$-BPS states forming a single centered black hole, by $d_*(m,n,\ell)$. Sen's conjecture \cite{Sen:2011ktd} then can be stated as
\begin{equation} \label{eq:positivity-conjecture}
     d_*(m,n,\ell) \, >\, 0\,,
\end{equation}
for the values $(m,n,\ell)$ for which there exists a single centered black hole, which satisfy $m,n\geq 1$ and $4mn-\ell^2>0$. Evidence for the conjecture \eqref{eq:positivity-conjecture} was given in the limit of large charges in \cite{Sen:2007qy, Dabholkar:2010rm}, for the first few values of the charge bilinears in \cite{Sen:2011ktd}, and for $m=1,2$ in \cite{Bringmann:2012zr}. The quantity $d_*(m,n,\ell)$ is defined in \eqref{eq:immortal-dyon-index} as the Fourier coefficient of a meromorphic Siegel modular form at a particular point in the Siegel upper-half space which is charge dependent. For generic values of $(m,n,\ell)$, it is not straightforwardly computed. However, as shown in Section \ref{sec2}, one can define the following mock Jacobi form \cite{Dabholkar:2012nd} which captures part of the spectrum with fixed charge bilinear $m$,
\begin{equation} \label{eq:psimf-intro}
    \psi_m^F(\sigma,v) \= \sum_{\substack{n,\ell\in \mathbb{Z} \\ n\geq -1}} c_m^F(n,\ell) q^n y^\ell\,, 
\end{equation}
where $q = e^{2\pi i \sigma}$  and $ y=e^{2\pi i v}$.  This provides an alternative route for computing $d_*(m,n,\ell)$ for a restricted set of charge bilinears $(m,n,\ell)$,
since for the ranges~$0 \leq \ell < 2m, 2n$, the generating function $\psi_m^F(\sigma,v)$ encodes the indexed degeneracies of the single centered black hole microstates  $d_*(m,n,\ell)$,
\begin{equation}
    d_*(m,n,\ell) \=  (-1)^{\ell+1}\,c^F_m(n,\ell) \,.
\end{equation}
Therefore, the conjecture \eqref{eq:positivity-conjecture} becomes a prediction for the sign of the coefficients of a family of mock Jacobi forms. Our main result is the following.
\begin{theorem} \label{Main-Theorem}
    Let $m>0$, $n\geq m$ and $0\leq \ell\leq m$, then
\begin{equation}
    (-1)^{\ell+1}\,c^F_m(n,\ell) \,>\,0 \,.
\end{equation}
\end{theorem}
The cases $m=1,2$ of this theorem were proved in \cite{Bringmann:2012zr}, so in this work we restrict ourselves to $m\geq 3$.  

As we explain in Section \ref{sec2}, as a consequence of the $S$-duality invariance of the single centered degeneracies, we have
\begin{corollary} \label{Main-Corollary}
    Let $m,n,\ell\in\mathbb{Z}$ with $m,n\geq1$ and $\Delta = 4mn-\ell^2>0$. Then,
\begin{equation}
    d_*(m,n,\ell) \,>\,0 \,.
\end{equation}
\end{corollary}
Or, as stated in the title, 

\begin{center}
    \textit{The index computing the single centered degeneracies of $\frac{1}{4}$-BPS dyons \\ in type II string theory compactified on $K3\times T^2$ is positive.}
\end{center}

The main idea behind the proof of the positivity conjecture for the single centered $\frac{1}{4}$-BPS states comes from the following physical intuition. The matching between the microscopic and macroscopic answers relies, so far, on the modularity properties of the index counting the BPS degeneracies, which allows for an analytic expression of the indexed degeneracies. This analytic expression, known as the Rademacher expansion, takes the form of a convergent series with a leading term and infinitely many suppressed corrections whose sum yields an integer. This expansion is then matched to the localization computation, where the path integral takes the form of a sum over saddles with the leading term corresponding to the near-horizon geometry of the black hole and the infinitely many exponentially suppressed contributions come from orbifolds of the near-horizon geometry \cite{Banerjee:2008ky,Banerjee:2009af,Murthy:2009dq,Dabholkar:2014ema,LopesCardoso:2022hvc} having the same boundary conditions as the unorbifolded one. The intuition is then that the saddle corresponding to the near-horizon geometry gives the dominant, positive contribution to the degeneracy which is not overshadowed by the remaining ones.

The Rademacher expansion for the coefficients $(-1)^{\ell+1}c_m^F(n,\ell)$ \cite{Bringmann:2010sd, Ferrari:2017msn, Chowdhury:2019mnb, LopesCardoso:2020pmp, Cardoso:2021gfg} is given in \eqref{eq:Rademacher-formula}. Following the previous discussion, the proof of Theorem \ref{Main-Theorem} simply consists in identifying the leading terms with a positive contribution and bounding the remaining infinitely many contributions against them.

We stress that the result in Corollary \ref{Main-Corollary} does not imply that all microstates are bosonic. It implies that, for each charge, the index receives more contributions from bosonic states than fermionic ones. To show that all microstates are bosonic, and furthermore that each of them carries zero angular momentum, one has to follow a different route, possibly adapting the results in \cite{Chowdhury:2014yca, Chowdhury:2015gbk, Maji:2023ims} for $\frac{1}{8}$-BPS states in $\mathcal{N}=8$ string theory to the case of $\frac{1}{4}$-BPS states in $\mathcal{N}=4$ string theory. 

The results of this work provide further evidence \cite{Sen:2011ktd} for the description of $\frac{1}{4}$-BPS states at strong coupling as black holes even for small charges. A further check would be to find a proof of the conjecture in \cite{Sen:2011ktd} for CHL models \cite{Chaudhuri:1995fk}, for which one could use the results about the sign of the index and the hair modes \cite{Chattopadhyaya:2017ews, Chattopadhyaya:2018xvg, Chattopadhyaya:2020yzl, Govindarajan:2022nzb, Chattopadhyaya:2023uov}.

\bigskip

The organization of the paper is as follows. In Section \ref{sec2} we define the index and give its expression in terms of the Fourier coefficients of a meromorphic Siegel modular form. We review the connection to mock modular forms and provide the exact expression for the indexed degeneracies with all the relevant details for the proof, including some unknown aspects of the formulae. In Section \ref{sec3 proof} we prove several small lemmas and combine them for the proof of Theorem \ref{Main-Theorem} for the cases $m\geq 6$. The proof of the cases $m=3,4,5$ is given in Appendix \ref{appendixA}. %
 
\section{Black hole degeneracies in \texorpdfstring{$\mathcal{N} = 4$}{N = 4}  string theory} \label{sec2}

\bigskip

The study of the BPS spectrum in string theory vacua with extended supersymmetry has uncovered rich mathematical structures. These provide non-trivial tests for dualities, as well as powerful tools to study black hole entropy at zero temperature. For the case of type II string theory compactified on $K3\times T^2$ (and supersymmetry preserving orbifolds, known as CHL models \cite{Chaudhuri:1995fk}), the study of its BPS spectrum has led to new findings in mathematics related to mock modularity \cite{Dabholkar:2012nd} and is deeply related to modern forms of moonshine (see \cite{Harrison:2022zee} for a recent review). In the next section we introduce the definition of the index counting $\frac{1}{4}$-BPS states for type II string theory compactified on $K3\times T^2$ and refer the reader to the reviews \cite{Sen:2007qy, Cheng:2008gx, Mandal:2010cj, Dabholkar:2012zz} for more details.

\subsection{Index definition}

Heterotic string theory compactified on $T^6$ (or, equivalently, Type II string theory compactified on $K3\times T^2$) has 28 gauge fields, 12 of them arising from the components of the metric and 2-form field on $T^6$ and 16 from the Cartan generators of the heterotic gauge group. The electric and magnetic charges of a generic state are described by the 28-dimensional vectors $\vec{Q}$ and $\vec{P}$, respectively. 

The vectors $\vec{Q},\vec{P}$ transform in the vector representation of the $T$-duality group of the theory $O(6,22;\mathbb{Z})$ and therefore  with an $O(6,22;\mathbb{Z})$-invariant matrix one can define the charge bilinears $\vec{Q}^2, \vec{P}^2, \vec{Q}\cdot \vec{P}$, which are invariant under $T$-duality transformations. The following integer 
\begin{equation}\label{eq:torsion}
   I \= \gcd \{ Q_i P_j - Q_j P_i\, , \quad 1\leq i,j \leq 28 \} \,,
\end{equation}
is called the torsion of the state carrying charges $(\vec{Q},\vec{P})$. A state with unit torsion,  $I=1$, 
is completely specified by the charge bilinears, $\vec{Q}^2, \vec{P}^2, \vec{Q}\cdot \vec{P}$ \cite{Banerjee:2007sr}. From now on we focus on this class of states, since the degeneracies of dyons with torsion greater than one are related to the degeneracies of unit torsion dyons \cite{Banerjee:2008ri, Banerjee:2008pu, Dabholkar:2008zy}.

The $U$-duality group also includes the $SL(2,\mathbb{Z})$ group which acts on the doublet $(\vec{Q},\vec{P})$ by matrix multiplication. This group is the electromagnetic duality group on the heterotic frame, and the modular group of the $T^2$ factor in the type II frame, so that the full duality group contains the following group
\begin{equation} \label{eq:duality-group}
    SL(2,\mathbb{Z})\times O(6,22;\mathbb{Z}) \,.
\end{equation}
To simplify the notation, it is standard to define
\begin{equation}
    m \= \frac{\vec{P}^2}{2}\,, \qquad\qquad n \= \frac{\vec{Q}^2}{2}\,, \qquad\qquad \ell \= \vec{Q}\cdot \vec{P}\,,
\end{equation}
which are integers since $\vec{Q}^2, \vec{P}^2 \in 2\mathbb{Z}$ with the appropriate choice of inner product, and their associated discriminant is given by
\begin{equation}
    \Delta \= \vec{Q}^2 \vec{P}^2 - (\vec{Q}\cdot \vec{P})^2 \= 4mn-\ell^2 \,,
\end{equation}
which is invariant under the duality transformations \eqref{eq:duality-group}.

The $\frac{1}{2}$-BPS states have charge vectors $\vec{Q}, \vec{P}$ which are parallel and therefore there is a duality frame in which one of them, say $\vec{P}$, can be set to zero. The partition function for the degeneracies of $\frac{1}{2}$-BPS states with charge bilinear $n$ is given by \cite{Dabholkar:1989jt},
\begin{equation} \label{eq:half-bps}
    \frac{1}{\eta^{24}(\sigma)} \= q^{-1}\prod_{n=1}^{\infty} (1-q^n)^{-24} \= \sum_{n=-1}^{\infty} d(n) q^n\,,
\end{equation}
where $q= e^{2\pi i \sigma}$ with Im$(\sigma)>0$.

The indexed degeneracies of $\frac{1}{4}$-BPS states in heterotic string theory compactified on $T^6$ are computed through the sixth helicity trace index, defined as the following trace  \cite{Bachas:1996bp, Gregori:1997hi,Sen:2009vz}
\begin{equation} \label{eq:sixth-hel}
    B_6(\vec{Q}, \vec{P}) \= \frac{1}{6!} \mathrm{Tr}\, (-1)^F(2h)^6 \,,
\end{equation}
over states carrying charges $\vec{Q}, \vec{P}$, where the helicity $h$ is the third component of angular momentum in the rest frame. The $(2h)^6$ factor is introduced to soak up the fermion zero modes associated to the 12 broken supersymmetry generators. This factor guarantees that only states breaking 12 of the 16 supersymmetries contribute to the index (see \cite{Sen:2009vz, Dabholkar:2010rm} for details). The helicity trace index then takes the form
\begin{equation}
    B_6(\vec{Q},\vec{P}) \= \frac{1}{6!} \mathrm{Tr}\left( (-1)^F(2h)^6 \right) \= -d(m,n,\ell)   \,,
\end{equation}
where $d(m,n,\ell)$ is the Witten index over the remaining degrees of freedom once the zero modes have been factored out, producing the overall minus sign.  Assuming that there are no extra hair modes, $d(m,n,\ell)$ captures the indexed degeneracies of $\frac{1}{4}$-BPS states with charge bilinears $(m,n,\ell)$ which form single centered or two-centered black holes \cite{Dabholkar:2009dq}.

This index was conjectured \cite{Dijkgraaf:1996it} and derived and rederived \cite{David:2006yn, Shih:2005uc, Gaiotto:2005hc} to be equal to
\begin{equation} \label{eq:degen-igusa}
 d(m,n,\ell) \= \int_\mathcal{C} d\rho d\sigma d v  \, \frac{e^{-2\pi i \left(m \rho + n\sigma +\ell \left(v +\frac{1}{2}\right) \right) }}{\Phi_{10}(\rho,\sigma,v)} \,,
\end{equation}
where the integration contour $\mathcal{C}$ is taken to be 
\begin{equation} \label{eq:contour-real-part}
     0 \,\leq \,\rho_1, \sigma_1, v_1 \,\leq 1 \,, \qquad (\rho_1,\sigma_1,v_1) \= \mathrm{Re}(\rho,\sigma,v)\,, %
\end{equation}
with the imaginary parts $\rho_2,\sigma_2 >0$ and $v_2$ kept fixed and satisfying
\begin{equation} \label{eq:contour-imag-part}
    \rho_2\sigma_2 -v_2^2 \, \gg 0 \,.
\end{equation}
The function $\Phi_{10}(\rho,\sigma,v)$ is the Igusa cusp form, which is a Siegel modular form which transforms under $Sp(4,\mathbb{Z})$ transformations with weight 10, i.e., satisfying
\begin{equation}
    \Phi_{10}\left((A\Omega+B)(C\Omega+D)^{-1}\right) \= \det(C\Omega+D)^{10}\Phi_{10}\left(\Omega\right) \,,
\end{equation}
where 
\begin{equation}
    \Omega \= \begin{pmatrix}
        \rho & v \\ v & \sigma
    \end{pmatrix}\, \quad \text{ and } \quad \begin{pmatrix}
        A & B \\ C & D
    \end{pmatrix} \in Sp(4,\mathbb{Z}) \,,
\end{equation}
(see Appendix A in \cite{Bossard:2018rlt} for a useful compendium on Siegel modular forms). In order to define the Igusa cusp form, we first introduce the Jacobi form $\varphi_{0,1}(\sigma,v)$ of weight $0$ and index 1 under $SL(2,\mathbb{Z})\ltimes \mathbb{Z}^2$ transformations (see \cite{eichler2013theory} for an introduction to Jacobi forms)
\begin{equation} \label{eq:def-phi01-thetas}
    \varphi_{0,1}(\sigma,v) \= 4 \left( \frac{\vartheta_2^2(\sigma,v)}{\vartheta_2^2(\sigma,0)} +\frac{\vartheta_3^2(\sigma,v)}{\vartheta_3^2(\sigma,0)} + \frac{\vartheta_4^2(\sigma,v)}{\vartheta_4^2(\sigma,0)} \right)\,. %
\end{equation}
The functions ${\vartheta_i(\sigma,v)}$ for $i=1,2,3,4$ are the classical Jacobi theta functions,
\begin{equation*}
    \vartheta_{1}(\sigma,v) \= \sum_{n\in \mathbb{Z}} (-1)^n q^{\frac{1}{2}\left(n-\frac{1}{2} \right)^2}y^{n-\frac{1}{2}}\, , \qquad 
    \vartheta_{2}(\sigma,v) \= \sum_{n\in \mathbb{Z}}  q^{\frac{1}{2}\left(n-\frac{1}{2} \right)^2}y^{n-\frac{1}{2}}\,,
\end{equation*}
\begin{equation*}
    \vartheta_{3}(\sigma,v) \= \sum_{n\in \mathbb{Z}}  q^{\frac{n^2}{2}}y^{n}\, , \qquad 
    \vartheta_{4}(\sigma,v) \= \sum_{n\in \mathbb{Z}} (-1)^n  q^{\frac{n^2}{2}}y^{n} \,,
\end{equation*}
where
\begin{equation}
        q = e^{2\pi i \sigma} \, , \quad y=e^{2\pi i v} \,.
\end{equation}
The Jacobi form $\varphi_{0,1}(\sigma,v)$ has the following Fourier expansion,
\begin{equation}
    \varphi_{0,1}(\sigma,v) \= \sum_{\substack{n,\ell\in \mathbb{Z} \\ 4n-\ell^2\geq -1}} C_0(4n-\ell^2) q^n y^\ell\,,
\end{equation}
where the Fourier coefficients depend only on the combination $4n-\ell^2$ due to $\varphi_{0,1}(\sigma,v)$ being a Jacobi form of index 1. %
The Igusa cusp form can then be defined as the multiplicative (or Borcherds) lift of $2\varphi_{0,1}(\sigma,v)$ in the following way \cite{Gritsenko:1996kt,Gritsenko:1996tm},
\begin{equation} \label{eq:def-igusa-product}
    \Phi_{10}(\rho,\sigma,v) \= pqy\prod_{\substack{m,n,\ell \in \mathbb{Z} \\(m,n,\ell)>0}}(1-p^m q^n y^\ell)^{2C_0(4mn-\ell^2)} \,,
\end{equation}
where $(m,n,\ell)>0$ means that either $m > 0$ or $m = 0$ and $n > 0$, or $m = n = 0$ and $\ell < 0$. %

From the definition \eqref{eq:def-phi01-thetas} one can see that $\varphi_{0,1}(\sigma,0)=12$ and therefore
\begin{equation}
    \sum_{\ell} 2C_0(4mn-\ell^2) \= \begin{cases}
        24 \qquad \mathrm{ if } \quad m = 0 \quad \mathrm{ or } \quad n = 0 \,, \\
        0 \qquad \mathrm{ otherwise.}
    \end{cases}
\end{equation}
This, together with the value $C_0(-1)=1$ implies the following behaviour for $ \Phi_{10}(\rho,\sigma,v)$ near $v = 0$,
\begin{equation} \label{eq:phi10-igusa-v-near-0}
  {\Phi_{10}(\rho,\sigma,v)} \={(2\pi i v)^2}{\eta^{24} (\rho)}{ \eta^{24}(\sigma)}+\mathcal{O}(v^4)\,,
\end{equation}
which means that near $v=0$ the function  $ {\Phi_{10}(\rho,\sigma,v)} ^{-1}$ splits into a product of two $\frac{1}{2}$-BPS partition functions \eqref{eq:half-bps}.

The Igusa cusp form vanishes at $v=0$ and all its images under the group of $Sp(4,\mathbb{Z})$ transformations, implying that its reciprocal is a meromorphic function with infinitely many poles. The locus of the poles is parametrized by five integers $(m^2,m^1,j,n_2,n_1)$ \cite{Borcherds1995} and is given by an algebraic equation where the coefficients satisfy a constraint given by the norm of the vector of integers,
\begin{equation}
    m^2 - m^1 \rho + n_1 \sigma + n_2 ( \rho\sigma- v^2) + j v \=0\, , \qquad j^2 + 4 (m^1 n_1 + m^2 n_2) \= 1\,.
\end{equation}
The linear poles are the poles with $n_2=0$. They are the only ones compatible with the defining conditions \eqref{eq:contour-real-part} and \eqref{eq:contour-imag-part} of the contour $\mathcal{C}$. The imaginary part of the locus of these poles is parametrized by $SL(2,\mathbb{Z})$ matrices, and can be written as
\begin{equation} \label{eq:linear-pole}
    pq\sigma_2 + rs\rho_2 + (ps+qr)v_2 \= 0\, ,\quad \mathrm{ where }\quad \begin{pmatrix}
        p & q \\r & s
    \end{pmatrix}\in SL(2,\mathbb{Z})\,.
\end{equation}
By changing the imaginary parts of the variables in \eqref{eq:degen-igusa} one can cross one of these poles and therefore the value of the indexed degeneracies \eqref{eq:degen-igusa} depends on the integration contour $\mathcal{C}$ through the value of the imaginary parts of the integration variables $\rho,\sigma,v$. This phenomenon is an instance of wall-crossing. The indexed degeneracies jump discontinuously as one crosses a pole, with the jump being equal to the residue at the pole \cite{Sen:2007vb, Cheng:2007ch, Banerjee:2008yu}. The poles \eqref{eq:linear-pole} define planes in the space $(\rho_2,\sigma_2,v_2)$ which divide it in chambers where the indexed degeneracy is constant. The residues for the linear poles can be computed from the behaviour of $\Phi_{10}^{-1}$ near $v=0$, given in \eqref{eq:phi10-igusa-v-near-0}, 
together with the $SL(2,\mathbb{Z})$ invariance of $\Phi_{10}(\rho,\sigma,v)$ under the diagonal $SL(2,\mathbb{Z})$ subgroup of  $Sp(4,\mathbb{Z})$. For the linear pole \eqref{eq:linear-pole}, the residue of the integrand in \eqref{eq:degen-igusa} is
\begin{equation}
        \\
      (-1)^{\ell+1}( - 2 r s  {n}\,- 2 p q m\,+( p s + q r ){\ell} )
	 \, d(r^2  {n} +\, p^2m - p r {\ell})d( s^2 {n} + q^2m - q s {\ell}) \,,
\end{equation}
which is consistent with the change in the indexed degeneracies being due to the (dis)appearance of bound states of two $\frac{1}{2}$-BPS states \cite{Sen:2007vb, Cheng:2007ch}.

For each set of charges $m,n,\ell$ there is a special integration contour \cite{Cheng:2007ch} defined by
\begin{equation} \label{eq:immortal-contour}
    \mathcal{C}_* = \{ \, 0 \,\leq \,\rho_1, \sigma_1, v_1 \,\leq 1 \,, \quad  \rho_2 \= 2n\Lambda, \, \sigma_2 \= 2m\Lambda, \, v_2 \= -\ell\Lambda \, \}\,,
\end{equation}
where $\Lambda \gg 0$, such that no linear pole contributes to the index. This contour is called the attractor contour since it coincides with the one given by the values that the moduli take at the black hole horizon through the attractor mechanism \cite{Ferrara:1995ih}. With this choice of moduli, the degeneracy captured by the index is the single centered one. The immortal dyon index is then defined as
\begin{equation} \label{eq:immortal-dyon-index}
    d_*(m,n,\ell) \= \int_{\mathcal{C}_*} d\rho d\sigma d v  \, \frac{e^{-2\pi i \left(m \rho + n\sigma +\ell \left(v +\frac{1}{2}\right) \right) }}{\Phi_{10}(\rho,\sigma,v)} \,,
\end{equation}
which, following the discussion in the introduction, is conjectured to coincide with the degeneracy of states forming the black hole horizon and therefore it must be positive as summarized in \eqref{eq:positivity-conjecture}. 

It follows from the definition of the contour $\mathcal{C}_*$ \eqref{eq:immortal-contour} that the number of BPS indexed degeneracies associated to a single centered black hole given by \eqref{eq:immortal-dyon-index} is invariant under the duality group \eqref{eq:duality-group}. The discriminant $\Delta$ is also invariant under the duality transformations \eqref{eq:duality-group}, which leads to the following observations. The value of the single centered degeneracies depends only on the $SL(2,\mathbb{Z})$ equivalence class of the discriminant $\Delta$, while the inequivalent number of solutions is given by the class number associated to the discriminant of $(m,n,\ell)$ \cite{Moore:1998pn}. Roughly, the class number grows like $\sqrt{\Delta}$, which is the same growth for the area, and therefore the entropy, of the single centered black hole.

The coefficients of the Fourier expansion in $p \= e^{2\pi i \rho}$ of the Siegel modular form are two-variable functions which transform as Jacobi forms due to the inherited symmetries from the symplectic transformations preserving the Fourier expansion. That is, the coefficients in the Fourier expansion
\begin{equation}
    \frac{1}{\Phi_{10}(\rho,\sigma,v)} \= \sum_{m=-1}^\infty  \psi_m(\sigma,v)  p^m \,, 
\end{equation}
associated to $p^m$, which are the generating functions for the degeneracies for fixed charge bilinear $m$, are Jacobi forms of weight $-10$ and index $m$. This means that they satisfy the following transformations,
\begin{equation} \label{eq:jacobi-form-trans}
\begin{split}
       \psi_m\left( \frac{a\sigma+b}{c\sigma+d}, \frac{v}{c\sigma+d} \right)  & \= (c\sigma+d)^{-10} e^{2\pi i m \frac{c^2 v}{c\sigma+d}} \, \psi_m(\sigma,v) \, , \quad 
       \begin{pmatrix}
           a & b \\ c & d 
       \end{pmatrix}\in SL(2,\mathbb{Z})
       \, , \\
   \psi_m\left( \sigma, v+\lambda\sigma+\mu \right) & \= e^{-2\pi i m (\lambda^2\sigma + 2\lambda  v)}  \,\psi_m(\sigma,v) \, , \quad \lambda, \mu \in \mathbb{Z} \, .
\end{split}
\end{equation}
The relation between $\Phi_{10}^{-1}$ and mock Jacobi forms arises by writing the Jacobi forms in the Fourier expansion of $\Phi_{10}^{-1}$ as a sum of a meromorphic function and a non-meromorphic one in the following canonical form \cite{Dabholkar:2012nd},
\begin{equation}
    \frac{1}{\Phi_{10}(\rho,\sigma,v)} \= \sum_{m=-1}^\infty \left( \psi_m^F(\sigma,v)+ \psi_m^{P}(\sigma,v) \right) p^m \,,
\end{equation}
where the function $\psi_m^{P}(\sigma,v)$, called the polar, is built out of the poles of $\Phi_{10}(\rho,\sigma,v)^{-1}$ in the $(\sigma,v)$ plane,
\begin{equation} \label{eq:polarpart-Appell}
    \psi_m^{P}(\sigma,v) \= \frac{d(m)}{\eta^{24}(\sigma)}\mathcal{A}_{2,m}(\sigma,v) \= \frac{d(m)}{\eta^{24}(\sigma)}\sum_{s\in\mathbb{Z}}\frac{q^{ms^2+s}\,y^{2ms+1}}{(1-q^s y)^2} \,.
\end{equation}
The Appell-Lerch sum $\mathcal{A}_{2,m}(\sigma,v)$ is a mock Jabobi form \cite{zwegers2008mock}, which means that it transforms properly under modular transformations only after the addition of the non-holomorphic Eichler integral of another modular object, which is called the shadow of the mock Jacobi form. More precisely,
\begin{equation} \label{eq:completion-Appell}
    \widehat{\mathcal{A}}_{2,m}(\sigma,v) \= \mathcal{A}_{2,m}(\sigma,v) + \sqrt{\frac{m}{4\pi}} \sum_{\ell \in \mathbb{Z}/2m\mathbb{Z}}\left( \vartheta_{m,\ell}^0\right)^*(\sigma) \vartheta_{m,\ell} (\sigma,v) \,,
\end{equation}
transforms as a Jacobi form of weight 2 and index $m$, where $ \vartheta_{m,\ell} (\sigma,v)$ are the index $m$ theta functions
\begin{equation} \label{eq:general-theta}
    \vartheta_{m,\ell} (\sigma,v) \= \sum_{\substack{r \in \mathbb{Z}\\ r = \ell \mod 2m}} q^{\frac{r^2}{4m}}y^r \,,
\end{equation}
and $\left(\vartheta_{m,\ell}^0\right)^*(\sigma)$ denotes the non-holomorphic Eichler integral of the shadow $\vartheta^0_{m,\ell} (\sigma) = \vartheta_{m,\ell} (\sigma,0)$,
\begin{equation}
    \left(\vartheta_{m,\ell}^0\right)^*(\sigma) \= \left( \frac{i}{2\pi} \right)^{\frac{1}{2}}\int_{-\overline{\sigma}}^\infty (z+ \sigma)^{-\frac{3}{2}}\,\overline{\vartheta_{m,\ell}^0(-\overline{z})} \,dz\,.
\end{equation}

Since the sum $\psi_m^F(\sigma,v)+ \psi_m^{P}(\sigma,v)$ transforms as a Jacobi form, the function $\psi_m^F(\sigma,v)$ is also a mock Jacobi\footnote{The function $\psi_m^{P}(\sigma,v)$ is the product of a modular form and a mock Jacobi form and therefore being precise the functions $\psi_m^P(\sigma,v)$ and $\psi_m^F(\sigma,v)$ are mixed mock Jacobi forms.} form with shadow having the opposite sign that for $\psi_m^{P}(\sigma,v)$. The mock Jacobi form $\psi_m^F(\sigma,v)$ is called the finite part and since it is holomorphic, it has the following well-defined Fourier expansion
\begin{equation} \label{eq:finite-part}
    \psi_m^F(\sigma,v) \= \sum_{\substack{n,\ell\in \mathbb{Z} \\ n\geq -1}} c_m^F(n,\ell) q^n y^\ell\,,
\end{equation}
having no wall-crossing in the $(\sigma,v)$ plane. Since $\psi_m^F(\sigma,v)$ is a holomorphic function which satisfies the second transformation in \eqref{eq:jacobi-form-trans}, it has the following theta expansion,
\begin{equation}
	\psi_m^F(\sigma,v) \= \sum_{\ell\in \mathbb{Z}/2m\mathbb{Z}} f_{m,\ell}(\sigma) \vartheta_{m,\ell} (\sigma,v)\,,
\end{equation}
where $f_{m,\ell}(\sigma)$ are (mixed) mock modular forms. Thus, the coefficients $c_m^F(n,\ell)$ can be seen as the Fourier coefficients of the mixed mock modular forms $f_{m,\ell}(\sigma)$ \cite{Dabholkar:2012nd}.

\subsection{Single centered degeneracies}

The finite part $\psi_m^F$ as written in \eqref{eq:finite-part} is the generating function with fixed magnetic bilinear $m$ for the single centered (or immortal) indexed degeneracy of states carrying charge bilinears $m,n,\ell$ which satisfy $0\leq \ell < 2m,2n$ \cite{Dabholkar:2012nd}.

To see it, consider the chamber in the Siegel upper-half space bounded by the walls $v_2 =0$, $v_2 = -\rho_2$ and $v_2 = -\sigma_2$, known as the $\mathcal{R}$ chamber \cite{Sen:2011mh}. This chamber coincides with the attractor chamber for the charge bilinears which satisfy $0\leq \ell < 2m,2n$ \cite{Sen:2011ktd}. By performing the Fourier expansion of $\Phi_{10}(\rho,\sigma,v)^{-1}$ in $p=e^{2\pi i \rho}$, we have implicitly assumed that $\rho_2 \gg \sigma_2, \abs{v_2}$.
Expanding the Appell-Lerch sum appearing in \eqref{eq:polarpart-Appell},
\begin{equation} \label{eq:Appell-lerch-fourier}
    \mathcal{A}_{2,m}(\sigma,v)   \=  \frac{1}{2}\sum_{s,\ell\in\mathbb{Z}} \ell\, \left[ \mathrm{sign}(s+\frac{v_2}{\sigma_2}) + \mathrm{sign}\, \ell \right]\, q^{ms^2+\ell s}\,
 y^{2ms+\ell}\, ,
\end{equation}
we see that for $0 < -v_2/\sigma_2<1$, the Fourier coefficients of $\mathcal{A}_{2,m}(\sigma,v)$ with $0< \ell <2m$ vanish, and the $\ell=0$ Fourier coefficient is zero. Thus, in the $\mathcal{R}$ chamber, where $0 < -v_2/\sigma_2<1$, the indexed degeneracies \eqref{eq:degen-igusa} with $0\leq \ell <2m$ are given by the Fourier coefficients $c_m^F(n,\ell)$ of the finite part $\psi_m^F(\sigma,v)$. These coincide with the immortal degeneracies as long as $\ell < 2n$.

To further illustrate the necessity of these constraints, we provide an example with $\ell \geq 2n$ where the contribution of multi-centered black holes spoils the positivity of the degeneracies counted by $\psi_m^F(\sigma,v)$ (this was emphasized recently in \cite{Bhand:2023rhm} for the case $\ell\geq 2m$). Consider the set of charges $(m,n,\ell) = (10,1,6)$ with $\Delta = 4$ which satisfies $0\leq \ell<2m$ but does not satisfy the condition $\ell < 2n$. The corresponding Fourier coefficient of $\psi_m^F(\sigma,v)$ is
\begin{equation}
    c_{10}^F(1,6) \= 1393564656 \,,
\end{equation}
and satisfies $(-1)^{6+1}c_{10}^F(1,6)<0$. So by not imposing the condition $\ell < 2n$, the positivity conjecture in this case is not satisfied due to the contribution of multi-centered black holes to the index. The positivity is recovered by subtracting the contributions from multi-centered configurations. This can be done by subtracting the residues of $\Phi_{10}^{-1}$ at the poles one crosses when going from the attractor chamber for the charges $(m,n,\ell) = (10,1,6)$ to the $\mathcal{R}$ chamber. We obtain
\begin{equation}
    d_*(10,1,6) \= (-1)^{7}\left(1393564656 - 4 d(1) d(5) - 2d(1)d(2)\right) \= 50064\,,
\end{equation}
which is equal to value for the degeneracy of the representative of the equivalence class for which the attractor chamber is the $\mathcal{R}$ chamber, which is given by $(m,n,\ell) = (1, 1,0)$ and satisfies $(-1)c_1^F(1,0) = d_*(1, 1,0)=d_*(10,1,6)$.

\subsection{Exact expressions for the degeneracies}

The Fourier coefficients of the finite part $\psi_m^F(\sigma,v)$ possess an exact expression in the form of a generalized Rademacher expansion \cite{Ferrari:2017msn, Cardoso:2021gfg} due to the mixed-mock nature of the function $\psi_m^F(\sigma,v)$, 
\begin{equation} \label{eq:Rademacher-formula}
\begin{split}
    	& (-1)^{\ell+1}c_m^F(n,\ell) \=  \\
		& (-1)^{\ell} \, i^{1/2} \,  \sum_{\gamma=1}^{+\infty}\sum_{\tilde{\ell}\in\mathbb{Z}/2m\mathbb{Z}} \left( 2\pi \sum_{\substack{\tilde{n}\geq -1,\\
		\tilde{\Delta}<0}} c^F_m(\tilde{n},\tilde{\ell})\frac{{\rm Kl}(\frac{\Delta}{4m},\frac{\tilde{\Delta}}{4m},\gamma,\psi)_{\ell\tilde{\ell}}}{\gamma}
	  \left(\frac{ \vert\tilde{\Delta} \vert}{\Delta} \right)^{23/4} I_{23/2}\left(\frac{\pi}{\gamma m}\sqrt{\Delta\vert\tilde{\Delta}\vert} \right) \right.
	 \\
	&\qquad \qquad   -\delta_{\tilde{\ell},0}\sqrt{2m}\, d(m) 
	\frac{{\rm Kl} ( \frac{\Delta}{4m}, -1 ;\gamma,\psi)_{\ell 0}}{\sqrt{\gamma}}\left(\frac{4m}{\Delta} \right)^6 I_{12}\left(\frac{2\pi}{\gamma}\sqrt{\frac{\Delta}{m}} \right) \\
	&\qquad \qquad +\frac{1}{2\pi} d(m) \sum_{\substack{ g \in \mathbb{Z}/2m\gamma \mathbb{Z} \\ g = \tilde{\ell} \text{ mod }2m}}
	\frac{{\rm Kl }( \frac{\Delta}{4m}, -1-\frac{g^2}{4m} ;\gamma,\psi)_{\ell\tilde{\ell}}}{\gamma^2} \\
& \qquad \qquad \left.   \left(\frac{4m}{\Delta} \right)^{25/4}\int_{-1/\sqrt{m}}^{1/\sqrt{m}}dx'  f_{\gamma,g,m}(x')  (1-mx'^2)^{25/4}
	I_{25/2}\left(\frac{2\pi}{\gamma\sqrt{m}}\sqrt{\Delta(1-mx'^2)}\right) \right), 
 \end{split}
\end{equation}
where the discriminants $\Delta, \tilde{\Delta}$ are given by
\begin{equation}
    \Delta \= 4mn-\ell^2 \, , \qquad \tilde{\Delta} \= 4m\tilde{n}-\tilde{\ell}^2\,,
\end{equation}
and $d(m)$ is the $m$-th coefficient of $\eta^{-24}(\sigma)$,
\begin{equation} \label{eq:dedekindeta24}
	\frac{1}{\eta^{24}(\sigma)} \= \frac{1}{q\prod\limits_{n=1}^{\infty}(1-q^n)^{24}}\ = \sum_{n=-1}^{+\infty} d(n) \, e^{2\pi i \sigma n}, \;\;\; q\=e^{2\pi i \sigma} \,.
\end{equation}
The generalized Kloosterman sum has the following expression,
\begin{equation}
	{\rm Kl}(\frac{\Delta}{4m},\frac{\tilde{\Delta}}{4m},\gamma,\psi)_{\ell\tilde{\ell}} \= 
 \sum_{\substack{0\leq -\delta <\gamma\\ (\delta,\gamma)=1, \alpha\delta = 1 \text{ mod } \gamma}}e^{2\pi i \left( \frac{\alpha}{\gamma}\frac{\tilde{\Delta}}{4m} +\frac{\delta}{\gamma}\frac{\Delta}{4m}\right)} \, {\psi} (\Gamma)_{\tilde{\ell}\ell} \,, 
\end{equation}
where $\Gamma$ denotes the $SL(2,\mathbb{Z})$ matrix associated to $\gamma$,
\begin{equation}
    \Gamma \= \begin{pmatrix}
	      \alpha & \beta \\ \gamma & \delta
	  \end{pmatrix} \in SL(2,\mathbb{Z})\,,
\end{equation}
and the multiplier system $ {\psi} (\Gamma)_{{\tilde \ell} \ell} $ is the one of the theta functions $\vartheta_{m,\ell}(\sigma,v)$  defined in \eqref{eq:general-theta}, which for $\gamma\geq1$ it has the following form \cite{Zagier1989},
\begin{equation}
	  {\psi} (\Gamma)_{{\tilde \ell} \ell} \= \frac{1}{\sqrt{2m\gamma i }}\sum_{T\in\mathbb{Z}/\gamma\mathbb{Z}} e^{2\pi i \left(\frac{\alpha}{\gamma}\frac{({\tilde \ell }-2mT)^2}{4m}-\frac{\ell ({\tilde \ell} -2mT)}{2m\gamma} +\frac{\delta}{\gamma}\frac{\ell^2}{4m}   \right)} \,.
\end{equation}
Finally, the function $f_{\gamma,g,m}(x') $ inside the integral is given by
\begin{equation}
	f_{\gamma,g,m}(x') \,= \sum_{\substack{p\in\mathbb{Z}\\2m\gamma p+ g\neq 0}}	 \frac{\gamma^2}{\left( x'-i  \gamma p -i\frac{g}{2m} \right)^2} = \left\{\begin{aligned}
	 	\frac{\pi^2}{\sinh^2\left( \frac{\pi x'}{\gamma} - \frac{\pi i g}{2m \gamma} \right)} &\;\; \text{ if}\;\;\; g \neq 0 \text{ mod } 2m\gamma\,, \\
	 	\frac{\pi^2}{\sinh^2\left( \frac{\pi x'}{\gamma} \right)} - \frac{\gamma^2}{x'^2} & \;\;\text{ if}\;\;\; g = 0 \text{ mod } 2m\gamma \,.
	 \end{aligned}\right.
\end{equation}
Two different proofs of this formula can be found in \cite{Ferrari:2017msn} and in \cite{Cardoso:2021gfg}. In the derivation of the latter, the following explicit formula is found for the coefficients $c^F_m(\tilde{n},\tilde{\ell} )$ with $\tilde{\Delta}<0$,
\begin{equation} \label{polarcoef}
	c^F_m(\tilde{n},\tilde{\ell} )=
			\sum_{\substack{r,s>0    \\ q \in \mathbb{Z}/ s \mathbb{Z}, \\ 
		ps - qr = 1 \\  0 \leq \frac{\tilde \ell}{2m} -\frac{q}{s} <  \frac{1}{rs }  }} 
	( - 2 r s  \tilde{n}\,- 2 p q m\,+( p s + q r )\tilde{\ell} )
	 \, d(r^2  \tilde{n} +\, p^2m - p r \tilde{\ell})d( s^2 \tilde{n} + q^2m - q s \tilde{\ell}) \,.
\end{equation}
We expand on the meaning and implications of this formula in the following subsection, which reviews known results and expands on some unknown aspects of the formula which are used later for the proof of the positivity conjecture.

\subsection{Polar coefficients}
\label{sec:pol-coef}
The coefficients $c^F_m(\tilde{n},\tilde{\ell})$ with $\tilde{\Delta} = 4m\tilde{n}-\tilde{\ell}^2<0$ appearing in \eqref{eq:Rademacher-formula} are called polar coefficients. From the expression for the Rademacher expansion \eqref{eq:Rademacher-formula}, one can see that, together with the $m$-th coefficient of $\eta^{-24}(\sigma)$, they completely determine the coefficients $c^F_m({n},{\ell})$ with ${\Delta} = 4m{n}-{\ell}^2>0$ of the mock Jacobi form $\psi_m^F(\sigma,v)$.

A weakly holomorphic modular form $f(\sigma)$ for $SL(2,\mathbb{Z})$ (see, e.g., \cite{Dabholkar:2012nd}) has by definition a Fourier expansion of the form
\begin{equation}
    f(\sigma) \= \sum_{n=n_0}^\infty a_n q^n \, , \qquad q = e^{2\pi i \sigma},
\end{equation}
where the finite number of coefficients $a_n$ with $n<0$ are called polar coefficients, since they are responsible for the exponential growth of $f(\sigma)$ as $\sigma \to i \infty$. That is, for the pole at the cusp $i \infty$. The Rademacher expansion for forms of negative weight consists in exploiting this behaviour to express the rest of the coefficients in terms of the polar ones (see, e.g., \cite{rademacher, apostol2012modular}). Thus, after finding the polar coefficients by some means, the remaining infinitely-many coefficients are determined through the Rademacher expansion. In general, given a set of weakly holomorphic modular forms, one has to compute the polar coefficients on a case-by-case basis. The case at hand is different in that the mock Jacobi forms indexed by $m$ arise from a single Siegel modular form, and therefore one might expect to find a general formula for the polar coefficients of this family. This is indeed the case.

The polar coefficients $c^F_m(\tilde{n},\tilde{\ell})$ were systematically studied in \cite{Chowdhury:2019mnb} following \cite{Sen:2011mh}, where it was shown that they are algorithmically computable and given by a finite sum involving $SL(2,\mathbb{Z})$ matrices and the Fourier coefficients of $\eta^{-24}(\sigma)$. The structure underlying the sum was elucidated in \cite{LopesCardoso:2020pmp}, where it was found that the $SL(2,\mathbb{Z})$ matrices appearing in the sum are the ones that are built out of the continued fraction expansion of $\frac{\tilde{\ell}}{2m}$. The same continued fraction structure emerged in \cite{Cardoso:2021gfg} in the computation of the degeneracies as an infinite sum over residues of $\Phi_{10}^{-1}$, yielding the expression \eqref{polarcoef}. This equation can be rewritten as 
\begin{equation} \label{polar-coef-formula}
	c_m^F(\tilde{n},\tilde{\ell})_{\tilde{\Delta}<0} \= \sum_{G \in W(m,\tilde{n},\tilde{\ell}) } \ell_G\, d(m_G)d(n_G)\,,
\end{equation}
where  $(m_G, n_G, \ell_G)$ are the images under $SL(2,\mathbb{Z})$ of the triplet $(m,\tilde{n},\tilde{\ell})$,
\begin{equation}
\begin{split}
 n_G & \= s^2 \tilde{n} + q^2m - q s \tilde{\ell} \, ,\\
 m_G & \= r^2  \tilde{n} +\, p^2m - p r \tilde{\ell} \,, \\
\ell_G & \=- 2 r s  \tilde{n}\,- 2 p q m\,+( p s + q r )\tilde{\ell} \,,
\end{split}
\end{equation}
and $W(m,\tilde{n},\tilde{\ell})$ is the subset of $SL(2,\mathbb{Z})$ consisting of matrices which satisfy the continued fraction condition for $\frac{\tilde{\ell}}{2m}$ with a non-zero contribution to $c_m^F(\tilde{n},\tilde{\ell})$,
\begin{equation} \label{eq:Wset-def}
\begin{split}
    & W(m,\tilde{n},\tilde{\ell}) \= \\
    & \left\{ G= \begin{pmatrix}
		p & q \\
		r & s
	\end{pmatrix} \in SL(2,\mathbb{Z}) : r,s>0, \, 0 \leq \frac{\tilde{\ell}}{2m} - \frac{q}{s} < \frac{1}{rs}, \; d(m_G),d(n_G) \geq -1   \right\}\,,
\end{split}
\end{equation}
where the continued fraction condition is given by the inequalities
\begin{equation} \label{eq:cont-fract}
     0 \,\leq \, \frac{\tilde{\ell}}{2m} - \frac{q}{s} \,< \,\frac{1}{rs} \,.
\end{equation}
We call these inequalities the continued fraction condition since they are satisfied by all the matrices built out of the successive convergents and semiconvergents of the continued fraction of the rational number $\frac{\tilde{\ell}}{2m}$. Given the continued fraction of $\frac{\tilde{\ell}}{2m}$,
\begin{equation}
    \frac{\tilde{\ell}}{2m} \= a_0+ \cfrac{1}{a_1 + \cfrac{1}{a_2+\cfrac{1}{\ddots + a_N} }} \= [a_0;a_1,\dots,a_N] \,,
\end{equation}
where $a_I \geq1$ for $I\geq 1$, the convergents are the rational numbers $[a_0;a_1,\dots,a_K]$ for $K\leq N$, while the semiconvergents are the numbers of the form $[a_0;a_1,\dots,a_{K-1},J]$ with $J\leq a_K$.

The computation of the polar coefficients is reduced to the characterization of the set $W(m,\tilde{n},\tilde{\ell})$. In the Rademacher formula \eqref{eq:Rademacher-formula} one is instructed to sum $\tilde{\ell}$ over a set of representatives of $\mathbb{Z}/2m\mathbb{Z}$. For simplicity we can take it to be $-m<\tilde{\ell}\leq m$. Since the mock Jacobi form has even weight, it satisfies the following identity
\begin{equation}
    \psi_m^F(\sigma,-v) \= \psi_m^F(\sigma,v)\,,
\end{equation}
which implies that
\begin{equation}
    c_m^F(n,-\ell) \= c_m^F(n,\ell)\,.
\end{equation}
Thus, we can reduce the range of $\tilde{\ell}$ in the polar coefficients $c_m^F(\tilde{n},\tilde{\ell})$ to $0\leq \tilde{\ell}\leq m$. For $\tilde{\ell}$ in this range, the set $W(m,\tilde{n},\tilde{\ell})$ contains the matrix
\begin{equation} \label{eq:umatrix}
    U \= \begin{pmatrix}
        1 & 0 \\ 1 & 1
    \end{pmatrix} \,,
\end{equation}
since the entries of the matrix $U$ satisfy the conditions \eqref{eq:cont-fract} and the contribution of this matrix to the polar coefficient $c_m^F(\tilde{n},\tilde{\ell})$ is given by
\begin{equation}
    \ell_U\, d(m_U)d(n_U) \= (\tilde{\ell}-2\tilde{n})d(m+\tilde{n}-\tilde{\ell})d(\tilde{n})\,,
\end{equation}
which is non-zero as long as $\tilde{n}\geq -1$ and $0\leq \tilde{\ell}\leq m$. Therefore, the set $W(m,\tilde{n},\tilde{\ell})$ is non-empty, as it always contains the $U$ matrix. We define the first contribution to the polar coefficient as
\begin{equation} \label{eq:first-contr}
    \hat{c}^F_m(\tilde{n},\tilde{\ell}) \= (\tilde{\ell}-2\tilde{n})d(m+\tilde{n}-\tilde{\ell})d(\tilde{n}) \,.
\end{equation}
This corresponds to the contribution from the leading pole to the degeneracies \cite{LopesCardoso:2004law, David:2006yn, Banerjee:2008ky, Cardoso:2021gfg}, or, equivalently, to the contribution from the AdS$_2\times S^2$ saddle to the gravitational path integral \cite{LopesCardoso:2022hvc}, as first shown in \cite{Murthy:2015zzy}. The other saddles, which are subleading, correspond to orbifolds of the near-horizon geometry \cite{Banerjee:2008ky, Murthy:2009dq, LopesCardoso:2022hvc}. 

The contribution from the other terms in $W(m,\tilde{n},\tilde{\ell})$ follows from the continued fraction description. Since $0\leq \tilde{\ell}\leq m$, the first coefficient of the continued fraction of $\frac{\tilde\ell}{2m}$ vanishes, $a_0 =0$. Ignoring the $a_0$ coefficient,\footnote{The matrix associated to $a_0$ is $T^{a_0}$, which does not satisfy the conditions in \eqref{eq:Wset-def}.} from the remaining coefficients $a_I$ one can build the following ordered set of $SL(2,\mathbb{Z})$ matrices
\begin{equation} \label{eq:building-matrices-contfrac}
    \begin{pmatrix}
		1 & 0 \\
		1 & 1
	\end{pmatrix},\, \dots,\, \begin{pmatrix}
		1 & 0 \\
		1 & 1
	\end{pmatrix}^{a_1},\, \begin{pmatrix}
		1 & 0 \\
		1 & 1
	\end{pmatrix}^{a_1}\begin{pmatrix}
		1 & 1 \\
		0 & 1
	\end{pmatrix}, \dots ,\,
 \begin{pmatrix}
		1 & 0 \\
		1 & 1
	\end{pmatrix}^{a_1}
 \begin{pmatrix}
		1 & 1 \\
		0 & 1
	\end{pmatrix}^{a_2}, \dots
\end{equation}
Note that the first matrix is the $U$ matrix \eqref{eq:umatrix}. In general, a matrix in the set \eqref{eq:building-matrices-contfrac} has the form
\begin{equation} \label{eq:matrices-cont-fract}
 \begin{pmatrix}
		1 & 0 \\
		1 & 1
	\end{pmatrix}^{a_1}
 \begin{pmatrix}
		1 & 1 \\
		0 & 1
	\end{pmatrix}^{a_2}
 \cdots
 \begin{pmatrix}
		1 & 1 \\
		0 & 1
	\end{pmatrix}^{a_K}
 \begin{pmatrix}
		1 & 0 \\
		1 & 1
	\end{pmatrix}^{J}
\end{equation}
with $K$ even,~$K \leq N$ and~$J\leq a_{K+1}$ (with~$a_{N+1} = 0$). There are~$\sum_{i=1}^N a_i \leq 2m$ matrices of this form, and all of them satisfy the continued fraction condition \eqref{eq:cont-fract}, furnishing all possible elements of the set $W(m,\tilde{n},\tilde{\ell})$ (see \cite{LopesCardoso:2020pmp} for a detailed description and possible subtleties when $\tilde{\ell}=m$, which do not affect the later analysis).

From this point of view, the reason why the saddles associated to the matrices which are not the $U$ matrix are subleading is that the other possible elements in the set $W(m,\tilde{n},\tilde{\ell})$ have smaller arguments in the coefficients of $\eta^{-24}(\sigma)$. Therefore, their contribution to the polar coefficient is exponentially suppressed with respect to the first one \eqref{eq:first-contr}. To see that the arguments of the coefficients of $\eta^{-24}(\sigma)$ decrease, first notice that we can write
\begin{equation}
    \frac{n_G}{s^2} \= \frac{\tilde{\Delta}}{4m} + m\left( \frac{\tilde{\ell}}{2m}-\frac{q}{s} \right)^2\, , \qquad
    \frac{m_G}{r^2} \= \frac{\tilde{\Delta}}{4m} + m\left( \frac{p}{r}-\frac{\tilde{\ell}}{2m} \right)^2\,.
\end{equation}
At each step in building the matrices of the form \eqref{eq:matrices-cont-fract}, either $r$ or $s$ increases, making the difference $\frac{p}{r}-\frac{\tilde{\ell}}{2m}$ or $\frac{\tilde{\ell}}{2m}-\frac{q}{s}$ smaller by the properties of the continued fraction. Since $\tilde{\Delta}<0$, when $r$ increases $m_G$ becomes smaller and when $s$ increases $n_G$ becomes smaller.
In other words, writing the pairs $(m_G, n_G)$ by $(m_i, n_i)$ with $i=1,\dots,i_\mathrm{max}$, where they are ordered by their appearance from the continued fraction, either $m_{i+1}$ or $n_{i+1}$ is strictly smaller than $m_i$ or $n_i$, respectively. Furthermore, the finiteness of the sum in \eqref{polarcoef} follows from the fact that the Fourier coefficients $d(r)$ vanish for $r<-1$. %

There is a more geometric way to see the decrease of the quantities $m_G, n_G$ at each step in the continued fraction algorithm. One can check that treating $n_G$ ($m_G$) as a function of $q,s$ (resp. $p,r$),
\begin{equation} \label{eq:derivatives-mg-ng}
   \left(p \frac{\partial}{\partial q}+ r\frac{\partial}{\partial s}  \right) n_G \= -\ell_G\,, \qquad
   \left( q \frac{\partial}{\partial p}+ s\frac{\partial}{\partial r}  \right) m_G \= -\ell_G\,.
\end{equation}
The first equation can be written as $(p,r)\cdot \nabla n_G(q,s) =-\ell_G$, so that the derivative of $n_G$ in the $(q,s)$ plane in the direction $(p,r)$ is given by $-\ell_G$.  For the matrices satisfying the continued fraction condition, one has
\begin{equation} \label{eq:lg-positive}
    0 \,\leq \, \frac{\tilde{\ell}}{2m} - \frac{q}{s} \,< \,\frac{1}{rs} \, \Rightarrow \, \ell_G >0\,.
\end{equation}
Thus, for the pairs $(p,r)$ and $(q,s)$ forming an $SL(2,\mathbb{Z})$ matrix and satisfying the continued fraction condition, $(p,r)\cdot \nabla n_G(q,s)$ and $(q,s)\cdot \nabla m_G(p,r)$  are always negative. Given the semiconvergents $\frac{p}{r}$ and $\frac{q}{s}$ of $\frac{\tilde{\ell}}{2m}$, the next semiconvergent in the continued fraction algorithm, which substitutes either $\frac{p}{r}$ or $\frac{q}{s}$ in the successive $SL(2,\mathbb{Z})$ matrix, is given by Farey addition
\begin{equation} \label{eq:farey-addition}
   \frac{p}{r} \oplus  \frac{q}{s} \= \frac{p+q}{r+s}\,.
\end{equation}
 Therefore, $(p,q)$ changes in the $(q,s)$ direction and $(p,r)$ changes in the $(q,s)$ direction, along which $m_G$ and $n_G$ are decreasing by equations \eqref{eq:derivatives-mg-ng} and \eqref{eq:lg-positive}.\footnote{There is  an even more visual way to see the decrease of $m_G$ and $n_G$ by using the topograph of a quadratic form, a visual representation of a quadratic form introduced by Conway in \cite{7afd335e-bde4-3897-9875-1f8b53717c0d}.  
 One can view the triplet $(m,\tilde{n},\tilde{\ell})$ with $\tilde{\Delta} = 4m\tilde{n}-\tilde{\ell}^2<0$ as defining the indefinite binary quadratic form $Q(x,y) = m x^2 -\tilde{\ell}xy+ \tilde{n}y^2$, since $(m,\tilde{n},\tilde{\ell})$ and $Q(x,y)$ transform in the same way under $SL(2,\mathbb{Z})$. It is shown in \cite{7afd335e-bde4-3897-9875-1f8b53717c0d} that there is a \textit{river} in the topograph connecting the zeros of $Q(x,y)$, given by $\frac{x}{y} = \frac{\tilde{\ell}\pm \sqrt{\vert\tilde{\Delta}\vert}}{2m}$, along which $Q(x,y)$ is periodic. Furthermore, there is a result, called the \textit{Climbing Lemma}, which states that as one moves away from the river, $Q(x,y)$ is strictly increasing or decreasing. The continued fraction of $\frac{\tilde{\ell}}{2m}$  defines a path in the topograph transverse to the river, from the positive to the negative side. Therefore, by the Climbing Lemma the values that the quadratic form takes along the path, which correspond to $m_G$ or $n_G$, are decreasing. In other words, the path along the topograph defined by the continued fraction of $\frac{\tilde{\ell}}{2m}$ goes downhill.}

We conclude this section by characterizing the possible values of the discriminant $\tilde{\Delta}<0$ and the coefficients with the most negative one. The Fourier coefficients $c_m^F(\tilde{n},\tilde{\ell})$ of the mock Jacobi form $\psi_m^F$, are only non-zero for $\tilde{n}\geq -1$. Since they also satisfy the elliptic symmetry
\begin{equation}
	c_m^F(\tilde{n},\tilde{\ell}) \= c_m^F(\tilde{n}+mT^2+T,\tilde{\ell}+2mT) \,,
\end{equation}
for any $T\in\mathbb{Z}$, this implies that they are non-zero only for
\begin{equation}
	\tilde{\Delta} \= 4m\tilde{n}-\tilde{\ell}^2 \,\geq\, -4m-m^2 \, .
\end{equation}
The polar coefficients $c_m^F(\tilde{n},\tilde{\ell})$ with $\tilde{\Delta}<0$ saturating the previous bound are the extreme polar coefficients. Their contribution to the Rademacher expansion is the one with the largest argument for the modified Bessel function of the first kind with index $23/2$. In the Rademacher formula \eqref{eq:Rademacher-formula}, one is instructed to sum $\tilde{\ell}$ over the equivalence class $\mathbb{Z}/2m\mathbb{Z}$. As earlier, for simplicity we can take it to be $-m<\tilde{\ell}\leq m$, and in this case the extreme polar coefficient is
\begin{equation} \label{eq:extremepolarvalue}
	c_m^F(-1,m) \= \hat{c}_m^F(-1,m) \= m+2\,,
\end{equation}
which indeed has discriminant $\tilde{\Delta} = -4m-m^2$. The value \eqref{eq:extremepolarvalue} is easily computed by noting that $W(m,-1,m)$ has only one element, $U$, since any further contribution must contain $d(r)$ with $r<-1$.

\section{Proof of positivity of the index}
\label{sec3 proof}
\subsection{Main idea}

 Since the single centered degeneracies are $S$-duality invariant, it is enough to prove the positivity for the representatives of one $S$-duality equivalence class. That is, it is enough proving it for states satisfying $n\geq m$, $0\leq \ell < 2m$, and then it is automatically proved for all the $SL(2,\mathbb{Z})$ images. Furthermore, since the invariance group can be enlarged to $GL(2,\mathbb{Z})$ (see, e.g. \cite{Bossard:2018rlt}, or notice that the Igusa cusp form has even weight), it is enough to prove it for charge configurations satisfying
\begin{equation} \label{eq:representatives}
	n \geq m, \hspace{8mm} 0 \leq \ell \leq m,
\end{equation}
which for $m\geq 1$ automatically satisfy $\Delta > 0$.

Thus, to prove that the indexed degeneracies of single centered $\frac{1}{4}$ BPS states in $\mathcal{N}=4$ string theory (given by \eqref{eq:immortal-dyon-index}) are positive, we use the Rademacher expansion \eqref{eq:Rademacher-formula} for the set of representatives defined by \eqref{eq:representatives}. The Rademacher expansion takes the form of a leading term plus exponentially suppressed corrections given by the sum over $\gamma$, which converge to an integer. Therefore, the crux of the proof consists in showing that the leading term is positive.%
\begin{figure}[h!]
    \centering
    \includegraphics[scale=0.8]{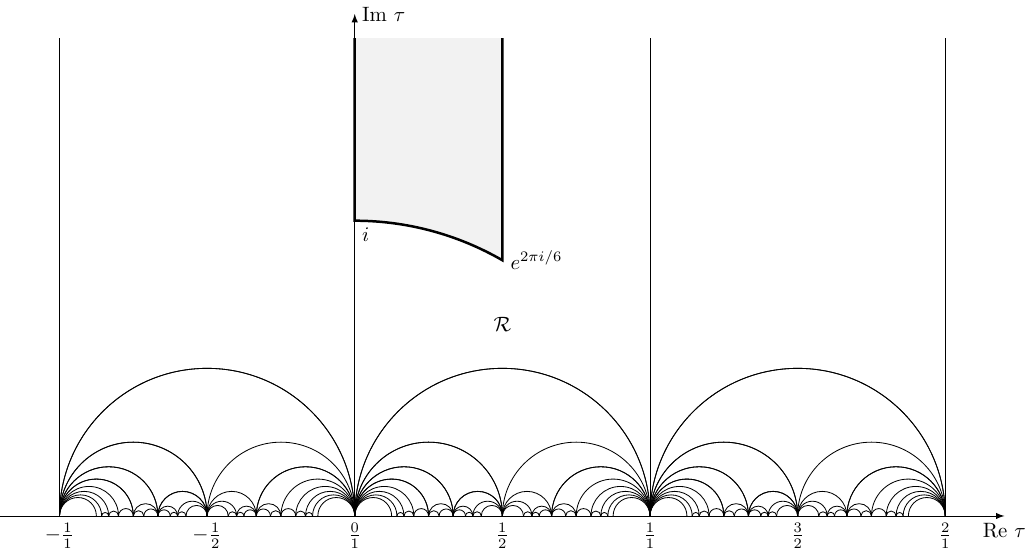}
    \caption{Walls of marginal stability in the upper half plane parametrized by $\tau =-\frac{v_2}{\sigma_2}+i \frac{\sqrt{\rho_2\sigma_2-v_2^2}}{\sigma_2}$, with $\rho_2\sigma_2-v_2^2\gg1$, which for the attractor values becomes $\tau = \frac{\ell}{2m}+i\frac{\sqrt{\Delta}}{2m}$. The $\mathcal{R}$ chamber is the one delimited by the walls connecting $0$, $1$ and $i\infty$. The delimited area in grey is a fundamental domain for $GL(2,\mathbb{Z})$, which lies inside the $\mathcal{R}$ chamber and contains the attractor values allowed by our choice of representatives \eqref{eq:representatives}.}
    \label{fig:enter-label}
\end{figure}

Our proof is similar to the proof of Bringmann-Murthy \cite{Bringmann:2012zr} for the case $m=2$. They show that the term given by the most polar state in the \textit{non-mock part} of the Rademacher expansion (the first line in the right hand side of \eqref{eq:Rademacher-formula}, as opposed to the second and third lines which are the \textit{mock part} since they originate from the anomalous modular transformation) gives the leading contribution to the degeneracy, being always positive and much larger than the rest of the terms combined. Instead, we identify several contributions which are positive, and their sum is greater than the rest combined. First, inspecting the non-mock part
\begin{equation*}
	 (-1)^{\ell} \, i^{1/2} \,2\pi  \sum_{\gamma=1}^{+\infty} \sum_{\substack{\tilde{\ell}\in\mathbb{Z}/2m\mathbb{Z}\\ \tilde{n}\geq -1,\\
		\tilde{\Delta}<0}} c^F_m(\tilde{n},\tilde{\ell})\frac{{\rm Kl}(\frac{\Delta}{4m},\frac{\tilde{\Delta}}{4m},\gamma,\psi)_{\ell\tilde{\ell}}}{\gamma}
	  \left(\frac{ \vert\tilde{\Delta} \vert}{\Delta} \right)^{23/4} I_{23/2}\left(\frac{\pi}{\gamma m}\sqrt{\Delta\vert\tilde{\Delta}\vert} \right) \,,
\end{equation*}
the leading term of the first sum is given by the $\gamma=1$ term, since all the other contributions are exponentially suppressed with respect to it. When $\gamma=1$, the associated $SL(2,\mathbb{Z})$ matrix is
\begin{equation}
	\begin{pmatrix}
		0 & -1 \\
		1 & 0
	\end{pmatrix}
\end{equation}
and the generalized Kloosterman sum gets simplified to
\begin{equation}
	{\rm Kl}(\frac{\Delta}{4m},\frac{\tilde{\Delta}}{4m},1,\psi)_{\ell \tilde{\ell}} \= \frac{e^{2\pi i \frac{\ell \tilde{\ell}}{2m}}}{\sqrt{2m i}} \,.
\end{equation} 
When $\tilde{\ell} = m$, we have
\begin{equation}
	{\rm Kl}(\frac{\Delta}{4m},\frac{\tilde{\Delta}}{4m},1,\psi)_{\ell m} \=  \frac{(-1)^{\ell}}{\sqrt{2m i}} \,.
\end{equation} 
Therefore, the $\gamma=1$ contribution to the Rademacher expansion for each of the terms with $\tilde{\ell} = m$ is positive. We write it as
\begin{equation} \label{eq:def-positive-contr}
   C_{\mathrm{pos.}}(\Delta, m, \tilde{n})\= \frac{2\pi}{\sqrt{2m}}  c^F_m(\tilde{n},m)
	  \left(\frac{  \vert\tilde{\Delta}_{m}\vert}{\Delta} \right)^{23/4} I_{23/2}\left(\frac{\pi}{ m}\sqrt{\Delta \vert\tilde{\Delta}_{m}\vert} \right) \,>\, 0\,,
\end{equation}
where we introduced the notation for the negative discriminant with $\tilde{\ell}=m$,
\begin{equation}
    \vert\tilde{\Delta}_{m}\vert \= 4m\tilde{n}-m^2\,,
\end{equation}
leaving the dependence on $\tilde{n}$ implicit.
From \eqref{eq:extremepolarvalue}, the expression \eqref{eq:def-positive-contr} for the extreme polar coefficient is
\begin{equation}
	 C_{\mathrm{pos.}}(\Delta, m, -1)\= 2\pi \, \frac{m+2}{\sqrt{2m}}
	  \left(\frac{ 4m+m^2}{\Delta} \right)^{23/4} I_{23/2}\left(2{\pi}\sqrt{\frac{\Delta}{m}\left(1+ \frac{m}{4} \right)} \right)\,.
\end{equation}
The argument of the Bessel function is larger than the arguments in any other Bessel function appearing in \eqref{eq:Rademacher-formula}. This term is responsible for the leading growth of $d(m,n,\ell)$ for $\Delta \gg 1$, correctly reproducing the Bekenstein-Hawking entropy $d(m,n,\ell)\sim e^{\pi\sqrt{\Delta}}$ since the area of the black holes with charge discriminant $\Delta$ is given by $4\pi\sqrt{\Delta}$ \cite{Cvetic:1995bj, Cvetic:1995uj}. In particular, over the course of the proof this term is shown to dominate over the mock-part of \eqref{eq:Rademacher-formula} given by terms with modified Bessel functions of the first kind of indices 12 and 25/2. %

An advantage of the choice of representatives given by \eqref{eq:representatives} is the following inequality for the discriminant,
\begin{equation} \label{eq:discriminant-bound-m2}
	\Delta \= 4mn-\ell^2 \,\geq\, 4m^2 - \ell^2 \geq 3m^2,
\end{equation} 
That is, the conditions $n\geq m$ and $0\leq \ell \leq m$ imply that $\Delta$ grows at least quadratically with $m$. This fact turns out to be crucial and will be used several times along the proof.

After bounding the mock part versus the non-mock part, the main part of the proof consists in showing that the estimates $e_{23/2}(\Delta,m,\tilde{n})$ defined by the following ratios
\begin{equation}
    e_{23/2}(\Delta,m,\tilde{n}) \=  \frac{1}{ C_{\mathrm{pos.}}(\Delta, m, \tilde{n})}\sum_{\substack{0\leq \tilde{\ell}< m \\ \tilde{\Delta}<0}} \frac{4\pi}{\sqrt{2m}}\, {c}_m^F(\tilde{n},\tilde{\ell}) \left(\frac{\vert\tilde{\Delta}\vert}{\Delta} \right)^{23/4}\, {I}_{23/2}\left(\frac{\pi}{m}\sqrt{\Delta \vert\tilde{\Delta}\vert}\right) \,,
\end{equation}
or, more explicitly, by,
\begin{equation}
    e_{23/2}(\Delta,m,\tilde{n}) \=  \sum_{\substack{0\leq \tilde{\ell}< m \\ \tilde{\Delta}<0}}\frac{2\, {c}_m^F(\tilde{n},\tilde{\ell}) \vert\tilde{\Delta}\vert^{23/4}\, {I}_{23/2}\left(\frac{\pi}{m}\sqrt{\Delta \vert\tilde{\Delta}\vert}\right)}
	{{c}_m^F(\tilde{n},m) \vert\tilde{\Delta}_{{m}}\vert^{23/4}\, {I}_{23/2}\left(\frac{\pi}{m}\sqrt{\Delta  \vert\tilde{\Delta}_{{m}}\vert}\right)} \,,
\end{equation}
are smaller than one.

\subsection{Notation}

To simplify the exposition we introduce the following notation
\begin{equation*} %
\begin{split}
   & C_{23/2}(m,n,\ell,\gamma) \= \\ & (-1)^{\ell} \, i^{1/2}  2\pi \sum_{\substack{\tilde{n}\geq -1,\\
		\tilde{\Delta}<0\\  \tilde{\ell}\in\mathbb{Z}/2m\mathbb{Z}}} c^F_m(\tilde{n},\tilde{\ell})
	  \left(\frac{ \vert\tilde{\Delta} \vert}{\Delta} \right)^{23/4}  
   \frac{{\rm Kl}(\frac{\Delta}{4m},\frac{\tilde{\Delta}}{4m},\gamma,\psi)_{\ell\tilde{\ell}}}{\gamma}
	  I_{23/2}\left(\frac{\pi}{\gamma m}\sqrt{\Delta\vert\tilde{\Delta}\vert} \right) \,,
\end{split}
\end{equation*}
\begin{equation*} %
	C_{12}(m,n,\ell,\gamma) \= (-1)^{\ell+1} \, i^{1/2} \sqrt{2m}\, d(m) 
	\frac{{\rm Kl} ( \frac{\Delta}{4m}, -1 ;\gamma,\psi)_{\ell 0}}{\sqrt{\gamma}}\left(\frac{4m}{\Delta} \right)^6 I_{12}\left(\frac{2\pi}{\gamma}\sqrt{\frac{\Delta}{m}} \right)\,,
\end{equation*}
\begin{equation*} %
\begin{split}
	C_{25/2}(m,n,\ell,\gamma) \= &  (-1)^{\ell} \, i^{1/2} \frac{1}{2\pi} d(m) \sum_{\substack{ g \in \mathbb{Z}/2m\gamma \mathbb{Z} \\ g = \tilde{\ell} \text{ mod }2m\\  \tilde{\ell}\in\mathbb{Z}/2m\mathbb{Z}}}
	\frac{{\rm Kl }( \frac{\Delta}{4m}, -1-\frac{g^2}{4m} ;\gamma,\psi)_{\ell\tilde{\ell}}}{\gamma^2}  \left(\frac{4m}{\Delta} \right)^{25/4}\\
&  \int_{-1/\sqrt{m}}^{1/\sqrt{m}}dx'  f_{\gamma,g,m}(x')  (1-mx'^2)^{25/4}
	I_{25/2}\left(\frac{2\pi}{\gamma\sqrt{m}}\sqrt{\Delta(1-mx'^2)}\right)\,,
\end{split}
	\end{equation*}
so that
\begin{equation}
	c_m^F(n,\ell) \= \sum_{\gamma =1}^{+\infty} \left( C_{23/2}(m,n,\ell,\gamma)+C_{12}(m,n,\ell,\gamma)+C_{25/2}(m,n,\ell,\gamma)\right) \,.
\end{equation}
Furthermore, in what follows we will use the following definitions
    \begin{equation}
        \Delta \= 4m n - \ell^2 \,, \qquad \tilde{\Delta } \= 4m\tilde{n}-\tilde{\ell}^2 \, , \qquad \tilde{\Delta}_{{m}} \= 4 m \tilde{n}-m^2 \,.
    \end{equation}

\subsection{Preliminaries}

In this subsection we state results that are used several times in the following. The first one is a bound for ratios of modified Bessel functions of the first kind. 

\begin{proposition}
    For $0<x<y$ and $\nu \geq \frac{1}{2}$, we have
\begin{equation} \label{eq:ratio-bessels-bound}
	\frac{I_\nu(x)}{I_\nu(y)} \, < \, e^{x-y}\left(\frac{x}{y} \right)^{1/2}\,.
\end{equation}
\end{proposition}
The proof can be found in \cite{baricz_2010}. The following proposition is used to estimate the growth of the polar coefficients. 

The second result is a rough bound for the Fourier coefficients of $\eta^{-24}(\tau)$ which is used for studying the growth of the polar coefficients of $\psi_m^F(\sigma,v)$.
\begin{proposition} \label{prop:eta-bound}
    Let $n\in\mathbb{N}_0$, the Fourier coefficients $d(n)$ of $\eta^{-24}(\tau)$ satisfy
\begin{equation} \label{eq:dedekindeta-bounds}
    \left(1-\frac{1}{12} \right)2\pi n^{-13/2}I_{13}(4\pi\sqrt{n}) \,<\, d(n) \,<\, \left(1+\frac{1}{12} \right)2\pi n^{-13/2}I_{13}(4\pi\sqrt{n})\,.
\end{equation}
\end{proposition}
\begin{proof}
    The coefficients $d(n)$ with $n\geq 1$ have the following Rademacher expansion,
    \begin{equation} \label{eq:rademachereta-prop}
        d(n) \= 2\pi \sum_{c=1}^\infty \frac{\mathrm{Kl}(n,-1,c)}{c} n^{-13/2}I_{13}\left(\frac{4\pi}{c}\sqrt{n}\right)\,,
    \end{equation}
    where
    \begin{equation}
        \mathrm{Kl}(a,b,c) \= \sum_{\substack{ 0\leq x \leq c-1 \\ \gcd(c,x)=1}} e^{\frac{2\pi i}{c} \left( ax+bx^* \right)}\,,
    \end{equation}
    is the classical Kloosterman sum, with $x^*$ being the inverse of $x$ modulo $c$. The expression \eqref{eq:rademachereta-prop} is also defined for $n=0$ by taking the limit $n\to 0$ from above. The value of the Kloosterman sum for $c=1$ is simply $1$, while for $c\geq 1$ it  satisfies the following bound
    \begin{equation}
        \left\vert\frac{\mathrm{Kl}(n,-1,c)}{c}\right\vert \,\leq\, 1\,.
    \end{equation}
    Thus, we can split
    \begin{equation}
        d(n) \= 2\pi n^{-13/2}I_{13}\left({4\pi}\sqrt{n}\right) + 2\pi \sum_{c=2}^\infty \frac{\mathrm{Kl}(n,-1,c)}{c} n^{-13/2}I_{13}\left(\frac{4\pi}{c}\sqrt{n}\right)\,,
    \end{equation}
    and bound the absolute value of the series by
    \begin{equation}
       \left\vert2\pi \sum_{c=2}^\infty \frac{\mathrm{Kl}(n,-1,c)}{c} n^{-13/2}I_{13}\left(\frac{4\pi}{c}\sqrt{n}\right)\right\vert \,\leq\, 2\pi \sum_{c=2}^\infty  n^{-13/2}I_{13}\left(\frac{4\pi}{c}\sqrt{n}\right)\,.
    \end{equation}
    Since ${I_{\nu}(x)}$ and $\frac{I_{\nu}(x)}{x^{\nu}}$ are monotonically increasing functions of $x$ for $x, \nu>0$, we can bound
    \begin{equation}
    \begin{split}
       & 2\pi \sum_{c=2}^\infty  n^{-13/2}I_{13}\left(\frac{4\pi}{c}\sqrt{n}\right) \,\leq\, 2\pi \int_{1}^\infty  n^{-13/2}I_{13}\left(\frac{4\pi}{c}\sqrt{n}\right)dc \\
       & \,\leq\, 2\pi n^{-13/2}I_{13}\left({4\pi}\sqrt{n}\right) \int_{1}^\infty  c^{-13}dc \= \frac{1}{12}2\pi n^{-13/2}I_{13}\left({4\pi}\sqrt{n}\right). 
    \end{split}
    \end{equation}
\end{proof}

\subsection{Mock and higher \texorpdfstring{$\gamma$}{gamma} terms}

We start by bounding the higher $\gamma$ terms in the Rademacher expansion with respect to the $\gamma=1$ counterparts and finding bounds for the terms coming from the anomalous transformation due to the mock modular properties of the Jacobi forms, which as stated earlier are the last two lines in the right hand side of \eqref{eq:Rademacher-formula} and we call them the mock terms.

\begin{lemma} The following quantities appearing in the Rademacher expansion \eqref{eq:Rademacher-formula} satisfy the  bounds \label{lemma:bounds-besselterms-general}
\begin{equation} \label{eq:bound-I23-highergamma}
    \left\vert\sum_{\gamma=2}^\infty \frac{{\rm Kl}(\frac{\Delta}{4m},\frac{\tilde{\Delta}}{4m},\gamma,\psi)_{\ell\tilde{\ell}}}{\gamma}
	  I_{23/2}\left(\frac{\pi}{\gamma m}\sqrt{\Delta\vert\tilde{\Delta}\vert} \right) \right\vert \, < \, 
	  \frac{2}{21}I_{23/2}\left(\frac{\pi}{ m}\sqrt{\Delta\vert\tilde{\Delta}\vert} \right)\,,
\end{equation}
\begin{equation} \label{eq:bound-I12-highergamma}
	 	\left\vert \sum_{\gamma=2}^{+\infty} \frac{{\rm Kl} ( \frac{\Delta}{4m}, -1 ;\gamma,\psi)_{\ell 0}}{\sqrt{\gamma}}I_{12}\left(\frac{2\pi}{\gamma}\sqrt{\frac{\Delta}{m}} \right) \right\vert  \,<\, \frac{2}{21} I_{12}\left({2\pi}{}\sqrt{\frac{\Delta}{m}} \right) \,,
	 \end{equation}
  and
  \begin{equation} \label{eq:C25/2-proposition}
        \begin{split}
            &\left\vert \frac{1}{2\pi} d(m) \sum_{\gamma=1}^{+\infty}\sum_{\tilde{\ell}\in\mathbb{Z}/2m\mathbb{Z}} \sum_{\substack{ g \in \mathbb{Z}/2m\gamma \mathbb{Z} \\ g = \tilde{\ell} \text{ mod }2m}}\frac{{\rm Kl }( \frac{\Delta}{4m}, -1-\frac{g^2}{4m} ;\gamma,\psi)_{\ell\tilde{\ell}}}{\gamma^2} \left(\frac{4m}{\Delta} \right)^{25/4} \right.
	  \\
&  \left.  \int_{\frac{-1}{\sqrt{m}}}^{\frac{1}{\sqrt{m}}}dx'  f_{\gamma,g,m}(x')  (1-mx'^2)^{25/4}
	I_{25/2}\left(\frac{2\pi}{\gamma\sqrt{m}}\sqrt{\Delta(1-mx'^2)}\right) \right\vert  \\
 & \leq \, \left(1+\frac{2}{21}\right){d(m)}   \left(\frac{4m}{\Delta} \right)^{25/4} m^{3/2}\pi  \left(\frac{\pi^2}{3}+\frac{1}{3}\right) I_{25/2}\left(\frac{2\pi}{\sqrt{m}}\sqrt{\Delta}\right) \,.
\end{split}
\end{equation}
\end{lemma}
\begin{proof}
    Since ${ \psi} (\Gamma)_{{\tilde \ell} \ell} $ is a unitary matrix, it satisfies $\left\vert  { \psi} (\Gamma)_{{\tilde \ell} \ell} \right\vert \leq 1$, from which 
    \begin{equation} \label{eq:gen-kloosterman-bound}
	\left\vert {\rm Kl}(\frac{\Delta}{4m},\frac{\tilde{\Delta}}{4m},\gamma,\psi)_{\ell\tilde{\ell}} \right\vert \,\leq\, \gamma\,,
\end{equation}
follows. The proof of the first two bounds \eqref{eq:bound-I23-highergamma} and \eqref{eq:bound-I12-highergamma} is completed by following the same steps as in Proposition \ref{prop:eta-bound}. 

For the last bound \eqref{eq:C25/2-proposition}, we begin by finding a bound for the integral. Adapting the bounds in Lemma 3.2 of \cite{Bringmann:2010sd} to our function $f_{\gamma,g,m}(x')$, we have 
	  \begin{equation}
	  	\left\vert \sinh \left( \frac{\pi x'}{\gamma} - \frac{\pi i g}{2m \gamma} \right) \right\vert \= 	\left\vert \cosh \left( \frac{\pi x'}{\gamma} - \pi i\left(\frac{ g}{2m \gamma}+\frac{1}{2}\right) \right) \right\vert \,\geq\, 	\left\vert \sin \left(  \frac{\pi g}{2m \gamma} \right) \right\vert \,, 
	  \end{equation}
	  so for $-m\gamma <g \leq m\gamma$ and $g\neq 0$,
	 \begin{equation}
	 	\left\vert \frac{1}{ \sin \left(  \frac{\pi g}{2m \gamma} \right) }\right\vert \,\leq\, \frac{\pi^2}{4\left(\frac{\pi g}{2m \gamma} \right)^2 } \= \frac{\gamma^2 m^2}{g^2}\,.
	 \end{equation}   
	 For $g=0$ we can simply use the bound
	 \begin{equation}
	 	\left\vert \frac{1}{ \sinh^2(x) } - \frac{1}{x^2}\right\vert \,\leq\, \frac{1}{3}\,,
	 \end{equation}  
	 to obtain
	 \begin{equation}
	 	\left\vert  f_{\gamma,0,m}(x') \right\vert \,\leq\, \frac{\pi^2}{3} \, .
	 \end{equation}
	 Defining
	 \begin{equation} 
	 	h_{\gamma,g,m} \= \left\{\begin{aligned}
	 	\frac{ \gamma^2 m^2}{g^2} &\;\; \text{ if}\;\;\; g \neq 0 \text{ mod } 2m\gamma \\
	 	\frac{1}{3} & \;\;\text{ if}\;\;\; g = 0 \text{ mod } 2m\gamma \,,
	 \end{aligned}\right.
	 \end{equation}
	 we have shown
	 \begin{equation} \label{eq:bound-f-function}
	 	f_{\gamma,g,m}(x') \,\leq\, \pi^2 h_{\gamma,g,m} \, .
	 \end{equation}
	 The function $(1- m x'^2)$ is non-negative and bounded by 1 in the interval $[\frac{-1}{\sqrt{m}},\frac{1}{\sqrt{m}}]$, therefore
	 \begin{equation}
	 \begin{split}
	 		 &	\left\vert \left. \int_{-1/\sqrt{m}}^{1/\sqrt{m}}dx'  f_{\gamma,g,m}(x')  (1-mx'^2)^{25/4}
	  I_{25/2}\left(\frac{2\pi}{\gamma\sqrt{m}}\sqrt{\Delta(1-mx'^2)}\right) \right)  \right\vert  \,\leq\, \\
	  & \frac{2}{\sqrt{m}}\pi^2h_{\gamma,g,m}I_{25/2}\left(\frac{2\pi}{\gamma\sqrt{m}}\sqrt{\Delta}\right)\,.
	 \end{split}
	 \end{equation}

  Using the bound \eqref{eq:gen-kloosterman-bound}, we can combine the sums over $\tilde{\ell}$ and $g$ in \eqref{eq:C25/2-proposition} into a sum over $g \in \mathbb{Z}/2m\gamma \mathbb{Z}$. Moreover, from
	 \begin{equation} \label{eq:bound-sum-h-function}
	 	\sum_{g=1}^{m\gamma} h_{\gamma,g,m} \= \gamma^2m^2\sum_{g=1}^{m\gamma}\frac{1}{g^2} \,<\, \frac{\pi^2}{6}\gamma^2m^2\,,
	 \end{equation}
	 the sum over the remaining $g$ dependent part satisfies the bound
	 \begin{equation}
	 	\sum_{g = -m\gamma+1}^{m\gamma} h_{\gamma,g,m}  \,<\, \left(\frac{\pi^2}{3}+\frac{1}{3}\right)\gamma^2m^2 \,.
	 \end{equation}
	 The last part of the proof consists in bounding the sum over $\gamma \geq 2$, which can be done as in the previous cases.
\end{proof}

The following Lemma shows that the mock terms are suppressed with respect to the largest positive non-mock contribution---the one with $\tilde{n}=-1$.
\begin{lemma} \label{prop:I12-I25-total-bound}
Let $m\geq 6$, $n\geq m$ and $0\leq \ell\leq m$, $\Delta = 4mn-\ell^2$, then
\begin{equation}
     e_{12}(\Delta,m) \= \frac{\left\vert \sum\limits_{\gamma=1}^{+\infty} C_{12}(m,n,\ell,\gamma) \right\vert}{ C_{\mathrm{pos.}}(\Delta, m, -1)}
     \,< \,
     5\cdot 10^{-4}\,,
\end{equation}
and 
\begin{equation}
     e_{25/2}(\Delta,m) \= \frac{\left\vert \sum\limits_{\gamma=1}^{+\infty} C_{25/2}(m,n,\ell,\gamma) \right\vert}{C_{\mathrm{pos.}}(\Delta, m, -1)} 
     \,<\,
     0.015\,.
\end{equation}
\end{lemma}
\begin{proof}
    From \eqref{eq:bound-I12-highergamma}, we have
    \begin{equation}
         {\left\vert \sum\limits_{\gamma=1}^{+\infty} C_{12}(m,n,\ell,\gamma) \right\vert} \,<\, {\frac{23}{21}\sqrt{2m}\, d(m) 
	\left(\frac{4m}{\Delta} \right)^6 I_{12}\left( 2\pi \sqrt{\frac{\Delta}{m}} \right)}\,.
    \end{equation}
    Since $\Delta \geq 3m^2$ and $m\geq 6$, we have
    \begin{equation}
    \begin{split}
     &   \frac{{\frac{23}{21}\sqrt{2m}\, d(m) 
	\left(\frac{4m}{\Delta} \right)^6 I_{12}\left(2\pi \sqrt{\frac{\Delta}{m}}  \right)}}{ C_{\mathrm{pos.}}(\Delta, m, -1)} \,\leq\, 
    \frac{\frac{23}{21}\sqrt{2m}\, d(m) 
	\left(\frac{4m}{3m^2} \right)^6 I_{12}\left( 2\sqrt{3}\pi \sqrt{m} \right)}{C_{\mathrm{pos.}}(3m^2, m, -1)} 
   \\
  & \=\, 
    \frac{\frac{23}{21}\sqrt{2m}\, d(m) 
	\left(\frac{4m}{3m^2} \right)^6 I_{12}\left( 2\sqrt{3}\pi \sqrt{m} \right)}{2\pi \, (m+2)\frac{1}{\sqrt{2m}}
	  \left(\frac{  4m+m^2 }{3m^2} \right)^{23/4} I_{23/2}\left(\sqrt{3}{\pi}\sqrt{4m+m^2} \right)}\,.
    \end{split}
    \end{equation}
    The last expression is a decreasing function of $m$, being bounded by its value at $m=6$, which is smaller than $5\cdot 10^{-4}$. 
    
    The other bound is proved in the same way. From \eqref{eq:C25/2-proposition} and $\Delta \geq 3m^2$, we have
    \begin{equation*}
    \frac{\left\vert \sum\limits_{\gamma=1}^{+\infty} C_{25/2}(m,n,\ell,\gamma) \right\vert}{ C_{\mathrm{pos.}}(3m^2, m, -1)}  
   \,\leq\,      \frac{\left(1+\frac{2}{21}\right){d(m)}   \left(\frac{4m}{3m^2} \right)^{25/4} m^{3/2}\pi  \left(\frac{\pi^2}{3}+\frac{1}{3}\right) I_{25/2}\left({2\sqrt{3}\pi}{\sqrt{m}}\right)}{2\pi \, (m+2)\frac{1}{\sqrt{2m}}
	  \left(\frac{ 4m+m^2}{3m^2} \right)^{23/4} I_{23/2}\left(\sqrt{3}{\pi}\sqrt{4m+m^2} \right)}\,.
    \end{equation*}
    The last expression is a decreasing function of $m$, being bounded by its value at $m=6$ which is smaller than $0.015$.      
\end{proof}

\subsection{Polar coefficients bounds}

In this subsection we bound the polar coefficients $c_m^F(\tilde{n},\tilde{\ell})$ with $\tilde{\Delta} = 4m\tilde{n}-\tilde{\ell}^2<0$ with respect to the first contribution defined in \eqref{eq:first-contr} coming from the $U$ matrix. We start with the case $\tilde{n}=0$ for which there is an exact expression for the polar coefficients. This expression can be derived by noticing that the matrices of the form $U^r$ with $r\geq 1$ are enough to find a path from the $\mathcal{R}$ chamber to a zero degeneracy chamber. Alternatively, one can also compute the coefficients with $\tilde{n}=0$ as the $p^m y^{\tilde{\ell}}$ coefficient of the function $\frac{d(0)}{\eta^{24}(\rho)}\mathcal{A}_{2,0}(\rho,v)$. The explicit form is,
\begin{equation}
    c^F_m(0,\tilde{\ell})\=\tilde{\ell}\,d(0)\sum_{r=1}^{\lfloor \frac{m+1}{\tilde{\ell}} \rfloor}d(m-r\tilde{\ell}) \,,
\end{equation}
which is valid for $0<\tilde{\ell}\leq m$. From this expression we can prove the following. 
\begin{lemma} \label{lemma polarcoef ntilde0}
    Let $m\geq 6$ and $1\leq k < m$. Then
    \begin{equation}  \label{eq:polarbounds-ntilde0}
    {c}^F_m(0,m-k) \,\leq\, \left(1 + \frac{25}{d(k)} + \kappa_1 k^{1/2} \right)\hat{c}^F_m(0,m-k) \,,
    \end{equation}
    where $\kappa_1 = \frac{91}{88\pi}$.
\end{lemma}
\begin{proof}
The ratio vs. the first contribution reads
\begin{equation}
    \frac{c^F_m(0,\tilde{\ell})}{\hat{c}^F_m(0,\tilde{\ell})} 
    \= \frac{\sum\limits_{r=1}^{\lfloor \frac{m+1}{\tilde{\ell}} \rfloor}d(m-r\tilde{\ell})}{d(m-\tilde{\ell})} 
    \= \frac{1}{d(k)}\sum_{r=1}^{\lfloor \frac{m+1}{m-k} \rfloor}d(m-rm+rk) \, ,
\end{equation}
where $k=m-\tilde{\ell}$. There will be extra contributions starting at $r=2$ only if
\begin{equation}
    m-2m+2k \,\geq\, -1 \Longrightarrow k \,\geq\, \frac{m-1}{2} \,>\,2 \,,
\end{equation}
since we are only considering $m\geq 6$. Thus, the following equality is satisfied
\begin{equation}
    {c}^F_m(0,m-k) \= \hat{c}^F_m(0,m-k) \qquad \text{ for } k=1,2,3 \,.
\end{equation}
The following analysis is restricted to $k\geq 3$. For these cases, we have
\begin{equation}
   \frac{1}{d(k)}\sum_{r=1}^{\lfloor \frac{m+1}{m-k} \rfloor}d(m-rm+rk) \leq  \frac{1}{d(k)}\left( d(-1)+d(0)+\sum_{r=1}^{\lfloor \frac{m-1}{m-k} \rfloor}d(m-rm+rk) \right) \,.
\end{equation}
Using the second inequality in \eqref{eq:dedekindeta-bounds}, the discrete sum over $r$ can be bounded by an integral and we obtain
\begin{equation}
\begin{split}
        & \sum_{r=1}^{\lfloor \frac{m-1}{m-k} \rfloor}d(m-rm+rk) \, \leq\, \\
        & d(k)+\int_{1}^{ \frac{m-1}{m-k}} \left(1+\frac{1}{12}\right)2\pi (m-rm+rk)^{-13/2}I_{13}(4\pi\sqrt{m-rm+rk})dr \,,
\end{split}
\end{equation}
where we have used that the sum over $r$ is monotonically decreasing. The integral can be evaluated to yield the following bound
\begin{equation}
 \frac{1}{d(k)}\sum_{r=1}^{\lfloor \frac{m-1}{m-k} \rfloor}d(m-rm+rk) \leq   1+\frac{1}{d(k)}\left(1+\frac{1}{12}\right)\frac{2\pi}{m-k}\left( \frac{I_{12}(4\pi \sqrt{k})}{2k^6 \pi}-\frac{I_{12}(4\pi)}{2\pi} \right)  \,.
\end{equation}
Applying the first inequality in \eqref{eq:dedekindeta-bounds} to $d(k)$ in the denominator, together with the bound
\begin{equation}
    \frac{I_{12}(4\pi\sqrt{n})}{I_{13}(4\pi\sqrt{n})} \leq\frac{7}{4} \qquad \text{ for } n\geq 3\,,
\end{equation}
we obtain
\begin{equation}
   \sum_{r=1}^{\lfloor \frac{m+1}{m-k} \rfloor}d(m-rm+rk) \, \leq \, d(-1)+d(0)+d(k)+d(k)\frac{1+\frac{1}{12}}{1-\frac{1}{12}}k^{1/2}\frac{7}{8\pi}\,,
\end{equation}
where we have used the rough bound $(m-k)^{-1}\leq 1$. We can write this bound as
\begin{equation}
    \frac{{c}^F_m(0,m-k)}{\hat{c}^F_m(0,m-k)} \leq \frac{1}{d(k)}\left( d(-1)+d(0)+d(k)+d(k)\frac{1+\frac{1}{12}}{1-\frac{1}{12}}k^{1/2}\frac{7}{8\pi}\right) \,,
\end{equation}
which is equivalent to \eqref{eq:polarbounds-ntilde0}.
\end{proof}
We now bound the polar coefficients for $\tilde{n}\geq 1$. For this we use the general form of the polar coefficients \eqref{polar-coef-formula}.
\begin{lemma} \label{lemma:bounds-polar-coef-ngeq1}
    Let $m\geq 6$, $\tilde{n}\geq 1$, $1\leq k \leq m$ and $\tilde{\Delta} = 4m\tilde{n} - (m-k)^2<0$. Then
    \begin{equation} \label{eq:polarbounds-general-prop}
     c_m^F(\tilde{n},m-k) \,\leq\, \left(\kappa_2+ \kappa_3 k\right)\, \hat{c}_m^F(\tilde{n},m-k) \,,
\end{equation}
where $\kappa_2 = 181/100$ and $\kappa_3 = 243/400$.
\end{lemma}
\begin{proof}
    Consider first the second contribution to $W(m,\tilde{n},\tilde{\ell})$, where $\tilde{\ell} = m-k$. For all charge bilinears satisfying $0<\tilde{\ell}<m$, we have $0<\frac{\tilde{\ell}}{2m}<\frac{1}{2}$ and therefore the second matrix coming from the continued fraction expansion of $\frac{\tilde{\ell}}{2m}$ is always
\begin{equation}
    U^2 = \begin{pmatrix}
        1 & 0 \\ 2 & 1
    \end{pmatrix}, \;\;\;\; \text{ with } \;\;\; {\ell}_{U^2}d(m_{U^2})d({n}_{U^2}) \= (\tilde{\ell}-4\tilde{n})d(4\tilde{n}+m-2\tilde{\ell})d(\tilde{n}) \, .
\end{equation}
In terms of $k = m-\tilde{\ell}$ the contribution becomes 
\begin{equation}
    {\ell}_{U^2}d(m_{U^2})d(\tilde{n}_{U^2}) \= (m-k-4\tilde{n})d(4\tilde{n}-m+2k)d(\tilde{n}) \, .
\end{equation}
The ratio of the second contribution and the first contribution satisfies the following bound,
\begin{equation}
    \frac{(m-k-4\tilde{n})d(4\tilde{n}-m+2k)d(\tilde{n})}{(m-k-2\tilde{n})d(\tilde{n}+k)d(\tilde{n})} \,\leq\,  \frac{d(4\tilde{n}-m+2k)}{d(\tilde{n}+k)}\,,
\end{equation}
since $\tilde{n}\geq 1$. We use this inequality to bound the rest of the contributions. As explained in Section \ref{sec:pol-coef}, $m_G$ and ${n}_G$ are decreasing in terms of the sequence defined by the continued fraction. Moreover, since~${n}_U = \tilde{n} \geq 1$,~$m_U = \tilde{n}+k\geq 1$ and the discriminant~$ \tilde{\Delta}$ is invariant under~$SL(2,\mathbb{Z})$ transformations,
\begin{equation}
    \tilde{\Delta} \= 4m_U{n}_U-{\ell}_U^2 \= 4m_G {n}_G - {\ell}_G^2 \,,
\end{equation}
the value of ${\ell}_G$ cannot be greater than ${\ell}_U$,
\begin{equation}
    {\ell}_U \,\geq\, {\ell}_G\,.
\end{equation}
By the description of the set $W(m,\tilde{n},m-k)$ in terms of the continued fraction of $\frac{\tilde{\ell}}{2m}$ given in section \ref{sec:pol-coef}, there are less than $2m$ contributions to the polar coefficients, so we can bound the sum of all contributions by
\begin{equation} \label{eq:polarcoef-rough-sum}
    \sum_{G\in W(m,\tilde{n},m-k)\setminus \{U\} } \frac{{\ell}_G \, d(m_G) \, d({n}_G)}{{\ell}_U \, d(m_U) \, d({n}_U)} 
    \,<\, 
    2m \frac{d(4\tilde{n}-m+2k)}{d(\tilde{n}+k)}\,.
\end{equation}
The right hand side is decreasing with $m$. Thus, we can bound it by the smallest possible value of $m$. From the conditions
\begin{equation}
    \tilde{\Delta} = 4m\tilde{n}-(m-k)^2<0, \;\;\; 1\leq k <m \,,
\end{equation}
one finds 
\begin{equation} \label{eq:boundmpolar}
    m > k+2\tilde{n}+\sqrt{(k+2\tilde{n})^2-k^2} \,,
\end{equation}
and therefore
\begin{equation}
    2m \frac{d(4\tilde{n}-m+2k)}{d(\tilde{n}+k)} \leq  2
    \left( k+2\tilde{n}+\sqrt{(k+2\tilde{n})^2-k^2}\right) 
    \frac{d\left(\left\lfloor2\tilde{n}+k-\sqrt{(k+2\tilde{n})^2-k^2}\right\rfloor\right)}{d(\tilde{n}+k)} \,.
\end{equation}
This quantity is maximized by the value $\tilde{n}=1$,
\begin{equation} \label{eq:bounds-almost-polargeneral}
    2m \frac{d(4\tilde{n}-m+2k)}{d(\tilde{n}+k)} \,<\, (6k+8) \frac{d(\lfloor2+k-2\sqrt{k}\rfloor)}{d(1+k)} \,,
\end{equation}
and the ratio of the $\eta^{-24}(\tau)$ coefficients is maximized by its value at $k=1$,
\begin{equation}
    \frac{d(\lfloor2+k-2\sqrt{k}\rfloor)}{d(1+k)}\, \leq\, \frac{d(1)}{d(2)} \=  \frac{324}{3200}\,,
\end{equation}
yielding the desired bound
\begin{equation}
    \frac{1}{\hat{c}_m^F(m,\tilde{n},m-k)}\sum_{G\in W(m,\tilde{n},m-k)\setminus \{U\} } \tilde{\ell}_G \, d(m_G)d(\tilde{n}_G) 
    \,\leq \, 
    \frac{324}{3200}(6k+8) \,.
\end{equation}
\end{proof}
The previous result is a rough bound, mostly due to the use of the inequality \eqref{eq:polarcoef-rough-sum}. Numerical experiments show that there is no linear growth in $k$ and  presumably using the precise form of the polar coefficients \eqref{polar-coef-formula} one could prove a tighter bound. For our purposes, however, the bound is enough.

The remaining bounds concern the polar coefficients with $\tilde{n}=-1$. 
The coefficients $c_m^F(-1,\tilde{\ell})$ with $0\leq \tilde{\ell}\leq m$ have the following exact expression 
\begin{equation} \label{eq:n-minus-one-polar-coef}
    c_m^F(-1,\tilde{\ell}) \= \sum_{r\geq 1}(2r+\tilde{\ell}) d(m-r^2-r\tilde{\ell}) \,,
\end{equation}
where the sum over $r$ is bounded by
\begin{equation}
    r_{\max} \= {\lfloor\frac{-\tilde{\ell}+\sqrt{\tilde{\ell}^2+4m+4}}{2}\rfloor}\,,
\end{equation}
due to the vanishing of $d(n)$ for $n<-1$. The expression \eqref{eq:n-minus-one-polar-coef} can be obtained directly from the product representation of $\Phi_{10}(\rho,\sigma,v)$ \eqref{eq:def-igusa-product}, since for $\Phi_{10}(\rho,\sigma,v)^{-1}$  the $p^{-1}$ coefficient is given by the quantities $C_0(-1)=1$ and $C_0(0)=10$, yielding the following negative index Jacobi form
\begin{equation}
\begin{split}
        \psi_{-1}(\sigma,v) \= & \frac{q^{-1}y}{(1-y)^2}\prod_{n=1}^\infty (1-y q^n )^{-2}(1-y^{-1}q^n )^{-2}(1-q^n)^{-20} \\  \= & \frac{1}{\eta^{24}(\sigma)}\frac{1}{\varphi_{-2,1}(\sigma,v)}\,.
\end{split}
\end{equation}
The function $\varphi_{-2,1}(\sigma,v)^{-1}$ can be expressed as a derivative of an Appell-Lerch sum \cite{Bringmann:2015gla, Bossard:2018rlt}, giving the coefficients \eqref{eq:n-minus-one-polar-coef} \cite{LopesCardoso:2020pmp}.

In \eqref{eq:n-minus-one-polar-coef} the second contribution coming from $r=2$ is only non-vanishing if
\begin{equation}
    m-4-2\tilde{\ell} = m-4-2(m-k) \geq -1 \Rightarrow k \geq \frac{3+m}{2} \geq 5\,,
\end{equation}
since $m\geq 6$. Thus, in the following we only consider the case $k\geq 5$. We have

\begin{lemma} \label{prop:polar-nminus1}
    Let $m\geq 6$ and $5\leq k \leq m$, then 
    \begin{equation} \label{eq:bound-nminusone-polar}
    c_m^F(-1,m-k) \,\leq\, (\kappa_4+ \kappa_5 k^{1/2})\, \hat{c}_m^F(-1,m-k) \,,
\end{equation}
where the constants~$\kappa_4,\kappa_5$ are $\kappa_4 = 3+2\cdot10^{-5}$ and $\kappa_5 = \frac{1+\frac{1}{12}}{1-\frac{1}{12}}\frac{1}{4\pi}$.
\end{lemma}

\begin{proof}
    The proof is similar to the one for the $\tilde{n}=0$ case. We have
\begin{equation} \label{eq:first-step-nminus1}
    \frac{c_m^F(-1,m-k)}{\hat{c}_m^F(-1,m-k)} \,\leq\,   1 + \sum_{r\geq 2}\frac{(2r+m-k) d(m-r^2-r(m-k))}{(2+m-k)d(k-1)}\,. 
\end{equation}
We isolate the possible contribution containing the coefficient $d(-1)$ in order to use the bounds \eqref{eq:dedekindeta-bounds}. Denoting by $r_*$ the positive value of $r$ satisfying $m-r^2-r(m-k)=-1$, if $r_*$ is an integer, the corresponding summand in \eqref{eq:first-step-nminus1} can be bounded by
\begin{equation*}
   \frac{(2r_*+m-k) d(m-r_*^2-r_*(m-k))}{(2+m-k)d(k-1)} \= \frac{\sqrt{(m-k)^2+4m+4}}{(2+m-k)d(k-1)} \,\leq\,  \frac{\sqrt{1+k}}{d(k-1)} < 2\cdot 10^{-5} \,,
\end{equation*}
where the last bound follows from the ratio being decreasing with $k$, with the minimum value given by $k=5$.

It follows that the right hand side of \eqref{eq:first-step-nminus1} is less than
\begin{equation}
    1 + 2\cdot 10^{-5}+ \sum_{r\geq 2}^{r_{\max}'}\frac{(2r+m-k) d(m-r^2-r(m-k))}{(2+m-k)d(k-1)}\,,
\end{equation}
where
\begin{equation}
    r_{\max}'\=\left\lfloor\frac{-(m-k)+\sqrt{(m-k)^2+4m}}{2}\right\rfloor \,,
\end{equation}
 with the convention that the sum is zero if $r_{\max}'<2$. From the second inequality in \eqref{eq:dedekindeta-bounds}, the decreasing sum satisfies the following bound
\begin{equation}
\begin{split}
     &\sum_{r= 2}^{r_{\max}'} {(2r+m-k) d(m-r^2-r(m-k))} 
     \,\leq\,  {(4+m-k) d(m-4-2(m-k))}  \\
     &+ \int_2^{r_{max}'} dr \, (2r+m-k)\left(1+\frac{1}{12}\right)2\pi \frac{I_{13}(4\pi\sqrt{m-r^2-r(m-k)})}{(m-r^2-r(m-k))^{13/2}}\,.
\end{split}
\end{equation}
The integral can be evaluated using the following equality
\begin{equation*}
        \int_2^{r_{max}'} dr\, (2r+m-k) \frac{I_{13}(4\pi\sqrt{m-r^2-r(m-k)})}{(m-r^2-r(m-k))^{13/2}} %
        \= \frac{I_{12}(4\pi\sqrt{2k-m-4})}{2\pi(2k-m-4)^6} - \frac{ (2\pi)^{11}}{12!} \,,
\end{equation*}
which yields,
\begin{equation}
\begin{split}
        &\sum_{r\geq 2}^{r_{\max}'}\frac{(2r+m-k) d(m-r^2-r(m-k))}{(2+m-k)d(k-1)} \\
        < \, & \,\frac{(4+m-k)d(2k-m-4)}{(2+m-k)d(k-1)}+ \frac{\left(1+\frac{1}{12}\right)\frac{I_{12}(4\pi\sqrt{2k-m-4})}{(2k-m-4)^6}}{(2+m-k)d(k-1)} \,.
\end{split}
\end{equation}
The first summand is decreasing with $m$. Since for fixed $k$ the least possible value of $m$ is $k$, this term is bounded by $2$,%
\begin{equation}
         \frac{(4+m-k)d(2k-m-4)}{(2+m-k)d(k-1)} \, < \, 2 \,.
\end{equation}
Now we examine the second term, which is again decreasing with $m$ and bounded by the value $m=k$, 
\begin{equation}
    \frac{\left(1+\frac{1}{12}\right)\frac{I_{12}(4\pi\sqrt{2k-m-4})}{(2k-m-4)^6}}{(2+m-k)d(k-1)} 
    \,\leq\,
    \frac{\left(1+\frac{1}{12}\right)\frac{I_{12}(4\pi\sqrt{k-4})}{(k-4)^6}}{2d(k-1)} \,.
\end{equation}
We use the first inequality in \eqref{eq:dedekindeta-bounds} to write
\begin{equation}
    \frac{\left(1+\frac{1}{12}\right)\frac{I_{12}(4\pi\sqrt{k-4})}{(k-4)^6}}{2d(k-1)}
    \,\leq\,
    \frac{1+\frac{1}{12}}{1-\frac{1}{12}}\frac{(k-1)^{13/2}}{4\pi {(k-4)^6}}\frac{I_{12}(4\pi\sqrt{k-4})}{I_{13}(4\pi\sqrt{k-1})}\,.
\end{equation}
This last expression is an unbounded function of $k$ which is smaller than the function capturing its asymptotic growth,
\begin{equation}
    \frac{1+\frac{1}{12}}{1-\frac{1}{12}}\frac{(k-1)^{13/2}}{4\pi {(k-4)^6}}\frac{I_{12}(4\pi\sqrt{k-4})}{I_{13}(4\pi\sqrt{k-1})}
    \,\leq\, 
    \frac{1+\frac{1}{12}}{1-\frac{1}{12}}\frac{k^{1/2}}{4\pi} \,.
\end{equation}
Putting the previous bounds together, we have
\begin{equation}
     \frac{c_m^F(-1,m-k)}{\hat{c}_m^F(-1,m-k)} \,<\,
     3 + 2\cdot 10^{-5} + \frac{1+\frac{1}{12}}{1-\frac{1}{12}}\frac{k^{1/2}}{4\pi} \,.
\end{equation}
\end{proof}

\subsection{Leading term bounds}

The remaining bounds concern the the leading term in the Rademacher expansion, by which we mean the first line in the right hand side of \eqref{eq:Rademacher-formula}. We show that the positive contribution \eqref{eq:def-positive-contr} for each $\tilde{n}$ is larger than the sum of the remaining contributions with $\tilde{\ell}<m$. We first derive a bound for the cases $\tilde{n}\geq 0$ by considering only the first contribution to the polar coefficients \eqref{eq:first-contr} and then we correct it with the bounds for the polar coefficients derived in the previous subsection separately for $\tilde{n}=0$ and $\tilde{n}\geq 1$. The bounds for the more involved case $\tilde{n}=-1$ are derived at the end. 
\begin{lemma} \label{prop:first-contributions}
    Let $m\geq 6$, $n\geq m$, $0\leq \ell\leq m$, and $\tilde{n}\geq 0$. Then
    \begin{equation} \label{eq:thm:firstcontr}
           \sum_{\substack{0\leq \tilde{\ell}< m \\ \tilde{\Delta}<0}}\frac{2\, \hat{c}_m^F(\tilde{n},\tilde{\ell}) \vert\tilde{\Delta}\vert^{23/4}\, {I}_{23/2}\left(\frac{\pi}{m}\sqrt{\Delta \vert\tilde{\Delta}\vert}\right)}{{c}_m^F(\tilde{n},m) \vert\tilde{\Delta}_{{m}}\vert^{23/4}\, {I}_{23/2}\left(\frac{\pi}{m}\sqrt{\Delta  \vert\tilde{\Delta}_{{m}}\vert}\right)}   \, < \, 0.123 \,.
    \end{equation}
\end{lemma}
\begin{proof}
    Using the inequality \eqref{eq:discriminant-bound-m2}, $\Delta \geq 3m^2$, we can bound each summand in \eqref{eq:thm:firstcontr} by
    \begin{equation*}
	   \frac{2\, \hat{c}_m^F(\tilde{n},\tilde{\ell}) \vert\tilde{\Delta}\vert^{23/4}\, {I}_{23/2}\left(\frac{\pi}{m}\sqrt{\Delta \vert\tilde{\Delta}\vert}\right)}
	{{c}_m^F(\tilde{n},m) \vert\tilde{\Delta}_{{m}}\vert^{23/4}\, {I}_{23/2}\left(\frac{\pi}{m}\sqrt{\Delta  \vert\tilde{\Delta}_{{m}}\vert}\right)}\leq   \frac{2\, \hat{c}_m^F(\tilde{n},\tilde{\ell}) \vert\tilde{\Delta}\vert^{23/4}\, {I}_{23/2}\left(\sqrt{3}{\pi}\sqrt{ \vert\tilde{\Delta}\vert}\right)}
	{{c}_m^F(\tilde{n},m) \vert\tilde{\Delta}_{{m}}\vert^{23/4}\, {I}_{23/2}\left(\sqrt{3}{\pi}\sqrt{  \vert\tilde{\Delta}_{{m}}\vert}\right)}.
    \end{equation*}
    The right hand side of this equation has the following explicit expression
    \begin{equation*}
        \frac{2\, (m-k-2\tilde{n})d(k+\tilde{n})d(\tilde{n})((m-k)^2-4m\tilde{n})^{23/4}\, {I}_{23/2}\left(\sqrt{3}{\pi}\sqrt{ (m-k)^2-4m\tilde{n}}\right)}
	{(m-2\tilde{n})d(\tilde{n})d(\tilde{n})(m^2-4m\tilde{n})^{23/4}\, {I}_{23/2}\left(\sqrt{3}{\pi}\sqrt{  m^2-4m\tilde{n}}\right)}
    \end{equation*}
    where $k \= m-\tilde{\ell}$. %
    From the inequality \eqref{eq:ratio-bessels-bound} we obtain
    \begin{equation}
        \begin{split}
          &  \frac{2\, (m-k-2\tilde{n})d(k+\tilde{n})d(\tilde{n})((m-k)^2-4m\tilde{n})^{23/4}\, {I}_{23/2}\left(\sqrt{3}{\pi}\sqrt{ (m-k)^2-4m\tilde{n}}\right)}
	{(m-2\tilde{n})d(\tilde{n})d(\tilde{n})(m^2-4m\tilde{n})^{23/4}\, {I}_{23/2}\left(\sqrt{3}{\pi}\sqrt{  m^2-4m\tilde{n}}\right)} \\
  & <  \frac{2\, (m-k-2\tilde{n})d(k+\tilde{n})d(\tilde{n})\, }
	{(m-2\tilde{n})d(\tilde{n})d(\tilde{n})}
 \left(\frac{(m-k)^2-4m\tilde{n}}{m^2-4m\tilde{n}}\right)^6
 e^{\sqrt{3}{\pi}\left(\sqrt{ (m-k)^2-4m\tilde{n}}-\sqrt{  m^2-4m\tilde{n}}\right)}.
        \end{split}
    \end{equation}
    Since $\tilde{n}\geq 0$, the right hand side of the previous equation is an increasing function of $m$. Therefore, it is smaller or equal than its $m\to\infty$ limit, which is given by
    \begin{equation}
        \frac{2\,d(k+\tilde{n}) }
	{d(\tilde{n})}
 e^{-\sqrt{3}{\pi}k}.
    \end{equation}
    The sum over $\tilde{\ell}$ is restricted to the values which satisfy $\tilde{\Delta} = 4m\tilde{n}-\tilde{\ell}^2<0$, so the sum over $k$ is from $k=1$ to $k = \lfloor m-2\sqrt{m\tilde{n}}\rfloor$. Therefore, by taking the limit $m\to \infty$, we have shown
    \begin{equation}
         \sum_{\substack{0\leq \tilde{\ell}< m \\ \tilde{\Delta}<0}}\frac{2\, \hat{c}_m^F(\tilde{n},\tilde{\ell}) \vert\tilde{\Delta}\vert^{23/4}\, {I}_{23/2}\left(\frac{\pi}{m}\sqrt{\Delta \vert\tilde{\Delta}\vert}\right)}
	{{c}_m^F(\tilde{n},m) \vert\tilde{\Delta}_{{m}}\vert^{23/4}\, {I}_{23/2}\left(\frac{\pi}{m}\sqrt{\Delta  \vert\tilde{\Delta}_{{m}}\vert}\right)} 
 \,<\,
 \sum_{k=1}^\infty \frac{2\,d(k+\tilde{n}) }
	{d(\tilde{n})}
 e^{-\sqrt{3}{\pi}k} \,.
    \end{equation}
    The last expression is a decreasing function of $\tilde{n}$. We can write it as
    \begin{equation}
        \sum_{k=1}^\infty \frac{2\,d(k+\tilde{n}) }
	{d(\tilde{n})}
 e^{-\sqrt{3}{\pi}k} \= \frac{2e^{\sqrt{3}\pi \tilde{n}}}{d(\tilde{n})}\left(\eta^{-24}\left(\frac{i\sqrt{3}}{2} \right)-\sum_{r=-1}^{\tilde{n}}d(r)e^{-\sqrt{3}\pi r} \right) \,,
    \end{equation}
    which we can compute for any value of $\tilde{n}$. E.g., for $\tilde{n}=1$ the previous quantity is smaller than~$0.09$, and for
    $\tilde{n}=0$ it gives the desired bound
    \begin{equation}
	\frac{2}{24}\left(\eta^{-24}\left(\frac{i\sqrt{3}}{2} \right)-e^{\sqrt{3}\pi}-24 \right) \= 0.122 ... < 0.123 \,.
 \end{equation}
\end{proof}
The next step consists in using the bounds for the polar coefficients found in the previous subsection to find a bound for the complete contribution. \begin{lemma} \label{lemma:total-contr-generaln}
    Let $m\geq 6, n\geq m, 0\leq \ell\leq m$, and $\tilde{n}\geq 1$. Then
    \begin{equation} \label{eq:thm:firstcontr-generaln}
        e_{23/2}(\Delta,m,\tilde{n}) \= \sum_{\substack{0\leq \tilde{\ell}< m \\ \tilde{\Delta}<0}}\frac{2\, {c}_m^F(\tilde{n},\tilde{\ell}) \vert\tilde{\Delta}\vert^{23/4}\, {I}_{23/2}\left(\frac{\pi}{m}\sqrt{\Delta \vert\tilde{\Delta}\vert}\right)}
	{{c}_m^F(\tilde{n},m) \vert\tilde{\Delta}_{{m}}\vert^{23/4}\, {I}_{23/2}\left(\frac{\pi}{m}\sqrt{\Delta  \vert\tilde{\Delta}_{{m}}\vert}\right)} \,<\, 0.23 \,.
    \end{equation}
\end{lemma}
\begin{proof}
    Substituting ${c}_m^F(\tilde{n},m-k)$ by its bound from Lemma \ref{lemma:bounds-polar-coef-ngeq1}, $(\kappa_2+\kappa_3 k)\hat{c}_m^F(\tilde{n},m-k)$, following the steps in the proof of Lemma \ref{prop:first-contributions}, one arrives at
    \begin{equation} \label{eq:firsteq-bounds-general-n}
         \sum_{\substack{0\leq \tilde{\ell}< m \\ \tilde{\Delta}<0}}\frac{2\, (\kappa_2+\kappa_3 k)\hat{c}_m^F(\tilde{n},\tilde{\ell}) \vert\tilde{\Delta}\vert^{23/4}\, {I}_{23/2}\left(\frac{\pi}{m}\sqrt{\Delta \vert\tilde{\Delta}\vert}\right)}
	{{c}_m^F(\tilde{n},m) \vert\tilde{\Delta}_{{m}}\vert^{23/4}\, {I}_{23/2}\left(\frac{\pi}{m}\sqrt{\Delta  \vert\tilde{\Delta}_{{m}}\vert}\right)} \,<\, \sum_{k=1}^\infty \frac{2(\kappa_2+\kappa_3 k)\,d(k+\tilde{n}) }
	{d(\tilde{n})}
 e^{-\sqrt{3}{\pi}k} \,.
    \end{equation}
    We split the right hand side in two terms. The first one is proportional to the expression obtained in Lemma \ref{prop:first-contributions},
    \begin{equation}
        \sum_{k=1}^\infty \frac{2 \kappa_2\,d(k+\tilde{n}) }
	{d(\tilde{n})}
 e^{-\sqrt{3}{\pi}k} \= \frac{2 \kappa_2 e^{\sqrt{3}\pi \tilde{n}}}{d(\tilde{n})}\left(\eta^{-24}\left(\frac{i\sqrt{3}}{2} \right)-\sum_{r=-1}^{\tilde{n}}d(r)e^{-\sqrt{3}\pi r} \right) \,,
    \end{equation}
    which for $\tilde{n}=1$ has a value which is smaller than 0.17. The second term in \eqref{eq:firsteq-bounds-general-n} is also decreasing with $\tilde{n}$ and can be written as
    \begin{equation} \label{eq:final-bounds-ngeneral-eisensteinpart}
        \sum_{k=1}^\infty \frac{2\kappa_3 k\,d(k+\tilde{n}) }
	{d(\tilde{n})}
 e^{-\sqrt{3}{\pi}k} \= 
 \frac{2\kappa_3e^{\sqrt{3}\pi \tilde{n}}}{d(\tilde{n})}\left(-\frac{E_2\left(\frac{i\sqrt{3}}{2} \right)+\tilde{n}}{\eta^{24}\left(\frac{i\sqrt{3}}{2} \right)}-\sum_{r=-1}^{\tilde{n}}(r-\tilde{n})d(r)e^{-\sqrt{3}\pi r} \right) \,,
    \end{equation}
    in terms of the function
     \begin{equation}
	\frac{E_2(\tau)}{\eta^{24}(\tau)} \= -\sum_{n=-1}^\infty n\,d(n) e^{2\pi i \tau n} \,,
    \end{equation}
    where
    \begin{equation}
        E_2(\tau) \= \frac{1}{2\pi i} \frac{d}{d \tau} \log \eta^{24}(\tau) \= 1 -24\sum_{n=1}^\infty \sigma_1(n)e^{2\pi i \tau n}\,,
    \end{equation}
    is the Eisenstein series of weight 2 and $\sigma_1(n)$ is the sum-of-divisors function. The quantity \eqref{eq:final-bounds-ngeneral-eisensteinpart} is a decreasing function of $\tilde{n}$ with a value at $\tilde{n}=1$ which is smaller than $0.06$. The sum of the bounds for the two contributions in the right hand side of \eqref{eq:firsteq-bounds-general-n} yields the inequality \eqref{eq:thm:firstcontr-generaln}.
\end{proof}
\begin{lemma} \label{lemma:total-contr-n0}
    Let $m\geq 6, n\geq m, 0\leq \ell\leq m$, and $\tilde{n}= 0$. Then
    \begin{equation} \label{eq:thm:firstcontr-ntilde0}
        e_{23/2}(\Delta,m,0) \=  \sum_{\substack{0\leq \tilde{\ell}< m \\ \tilde{\Delta}<0}}\frac{2\, {c}_m^F(0,\tilde{\ell}) \vert\tilde{\Delta}\vert^{23/4}\, {I}_{23/2}\left(\frac{\pi}{m}\sqrt{\Delta \vert\tilde{\Delta}\vert}\right)}
	{{c}_m^F(0,m) \vert\tilde{\Delta}_{{m}}\vert^{23/4}\, {I}_{23/2}\left(\frac{\pi}{m}\sqrt{\Delta  \vert\tilde{\Delta}_{{m}}\vert}\right)} < 0.125 \,.
    \end{equation}
\end{lemma}
\begin{proof}
    The proof is essentially the same as for the previous Lemma, so we skip some details. Following the reasoning for Lemma \ref{prop:first-contributions} for the extra contribution due to the subleading terms in ${c}^F_m(0,m-k)$ given in Lemma 
    \ref{lemma polarcoef ntilde0}, which is only taken into account for $k\geq 3$, one can bound the extra contribution to $e_{23/2}(\Delta,m,0)$ by 
    \begin{equation}
    \begin{split}
                & \frac{2}{d(0)}\sum_{k=3}^\infty\left(25+\kappa_1 d(k) k^{1/2}\right) e^{-\sqrt{3}\pi k} \, < \, \\
                 & \frac{50}{d(0)}\sum_{k=3}^\infty  e^{-\sqrt{3}\pi k} + \frac{2\kappa_1}{d(0)} \left( -\frac{E_2\left(\frac{i\sqrt{3}}{2} \right)}{\eta^{24}\left(\frac{i\sqrt{3}}{2} \right)}-\sum_{r=-1}^{2}rd(r)e^{-\sqrt{3}\pi r}\right) 
                \, < \, 0.02 \,.
    \end{split}
    \end{equation}
\end{proof}

Now we turn to the final case $\tilde{n}=-1$. Equipped with the bound for the polar coefficients from the previous subsection, we prove the following result, which is the last preliminary result for the proof of the main theorem.
\begin{lemma} \label{lemma:total-contr-nminus1}
     Let $m\geq 6, n\geq m, 0\leq \ell\leq m$ and $\tilde{n}=-1$. Then
    \begin{equation} \label{eq:thm:nminus1}
        e_{23/2}(\Delta,m,-1) \= \sum_{0\leq \tilde{\ell}< m}\frac{2\, {c}_m^F(-1,\tilde{\ell}) \vert\tilde{\Delta}\vert^{23/4}\, {I}_{23/2}\left(\frac{\pi}{m}\sqrt{\Delta \vert\tilde{\Delta}\vert}\right)}
	{{c}_m^F(-1,m) \vert\tilde{\Delta}_{{m}}\vert^{23/4}\, {I}_{23/2}\left(\frac{\pi}{m}\sqrt{\Delta  \vert\tilde{\Delta}_{{m}}\vert}\right)} \,<\, 0.3 \,.
    \end{equation}
\end{lemma}
\begin{proof} In \eqref{eq:n-minus-one-polar-coef} the second contribution to the polar coefficient, given by $r=2$, is only non-vanishing if
\begin{equation}
    m-4-2\tilde{\ell} \= m-4-2(m-k) \,\geq\, -1 \quad\Longrightarrow\quad k \,\geq\, \frac{3+m}{2} \,\geq\, 5 \,,
\end{equation}
since $m\geq 6$. Thus, for the first four cases $k=1,2,3,4$, which are the ones with the largest contribution to \eqref{eq:thm:nminus1}, we can substitute ${c}_m^F(-1,\tilde{\ell})$ for $\hat{c}_m^F(-1,\tilde{\ell})$.  For  $k = m-\tilde{\ell}=1,2,3,4$, the summands are given by 
\begin{equation} \label{eq:ratio-nminus1}
\begin{split}
        & \frac{(m+2-k)\,d(k-1)\,2\, ((m-k)^2+4m)^{23/4}\, {I}_{23/2}\left(\sqrt{3}{\pi}\sqrt{(m-k)^2+4m} \right)}
	{(m+2)\,(m^2+4m)^{23/4} \, {I}_{23/2}\left(\sqrt{3}{\pi}\sqrt{m^2+4m}\right)} \\
\leq\, & 
     \,\frac{d(k-1)\,2\, ((m-k)^2+4m)^{23/4}\, {I}_{23/2}\left(\sqrt{3}{\pi}\sqrt{(m-k)^2+4m} \right)}
	{(m^2+4m)^{23/4} \, {I}_{23/2}\left(\sqrt{3}{\pi}\sqrt{m^2+4m}\right)} \,,
\end{split}
\end{equation}
where we have also used that, as a function of $\Delta$, the summands in \eqref{eq:thm:nminus1} are less or equal than their value for $\Delta = 3m^2$. Unlike the previous cases with $\tilde{n}>-1$, the ratio
\begin{equation}
     \frac{{I}_{23/2}\left(\sqrt{3}{\pi}\sqrt{(m-k)^2+4m} \right)}
	{{I}_{23/2}\left(\sqrt{3}{\pi}\sqrt{m^2+4m}\right)} \, , 
\end{equation}
is a decreasing function of $m$, and therefore we cannot bound \eqref{eq:ratio-nminus1} by its value at $m\to\infty$, which is $2d(k-1)e^{-\sqrt{3}\pi k}$. As a function of $m$, the right hand side of \eqref{eq:ratio-nminus1}
 is decreasing at $m=6$ and increasing for large enough $m$. Since for $m\geq 6$ it only has one local extremum, which is a minimum, \eqref{eq:ratio-nminus1} is bounded by its values at $m=6$ and by the limiting value $m\to \infty$. Computing the value of \eqref{eq:ratio-nminus1} for $m=6$ and $m\to\infty$, for all four cases $k=1,2,3,4$ we find that the maximum value is given by the value at $m=6$. These numerical values are smaller than $0.231, 0.022, 0.003, 6\cdot 10^{-4}$ for $k=1,2,3,4$, respectively. In other words, we have
 \begin{equation}
	\sum_{k=1}^4\frac{2\,c_m^F(-1,m-k) ((m-k))^2+4m)^{23/4}I_{23/2}\left(\frac{\pi}{m}\sqrt{\Delta((m-k))^2+4m)} \right)}
	{ c_m^F(-1,m)(m^2+4m)^{23/4}I_{23/2}\left(\frac{\pi}{m}\sqrt{\Delta(m^2+4m)}\right)} \,<\, 0.255
\end{equation}
for all $m\geq 6$.

We analyse the cases $k=5,6$ separately. For this we first compute the Fourier coefficients $c_m^F(-1,m-k)$ for $m=6$ and $k=5,6$,
\begin{equation*}
    c^F_6(-1,1) \= 3d(4)+5d(0) \= 52888, \;\;\; c^F_6(-1,0) \= 2d(5)+4d(2) \= 2160240\,.
\end{equation*}
With these we can bound $c^F_m(-1,m-5) < 1.0003\,\hat{c}^F_m(-1,m-5)$ and $c^F_m(-1,m-6) < 1.006\,\hat{c}^F_m(-1,m-6)$ for $m\geq 6$. To bound the quantities
\begin{equation} \label{eq:intermediate-bound-k56}
   \frac{2\,c_m^F(-1,m-k) ((m-k))^2+4m)^{23/4}I_{23/2}\left(\frac{\pi}{m}\sqrt{\Delta((m-k))^2+4m)} \right)}
	{ c_m^F(-1,m)(m^2+4m)^{23/4}I_{23/2}\left(\frac{\pi}{m}\sqrt{\Delta(m^2+4m)}\right)} \,,
\end{equation}
for $k=5,6$, we follow the same reasoning as before. Again, the ratio
\begin{equation}
    \frac{((m-k))^2+4m)^{23/4}I_{23/2}\left(\sqrt{3}{\pi}\sqrt{((m-k))^2+4m)} \right)}
	{(m^2+4m)^{23/4}I_{23/2}\left(\sqrt{3}{\pi}\sqrt{(m^2+4m)}\right)}\,,
\end{equation}
is initially decreasing with $m$ and asymptotes to $e^{-\sqrt{3}\pi k}$ for $m\to\infty$, having only one minimum in this range. Thus, we can bound \eqref{eq:intermediate-bound-k56} by the values at the extrema which turn out to be given by the values at $m=6$, which are smaller than $0.0012$ and $0.003$ for $k=5,6$ respectively. Thus, we have
 \begin{equation}
	\sum_{k=1}^6\frac{2\,c_m^F(-1,m-k) ((m-k)^2+4m)^{23/4}I_{23/2}\left(\frac{\pi}{m}\sqrt{\Delta((m-k)^2+4m)} \right)}
	{ c_m^F(-1,m)(m^2+4m)^{23/4}I_{23/2}\left(\frac{\pi}{m}\sqrt{\Delta(m^2+4m)}\right)}\, <\, 0.26
\end{equation}
for all $m\geq 6$.

Finally, we bound the remaining the contributions
\begin{equation} \label{eq:nminusone-othercases}
   \sum_{k=7}^{m}\frac{2\,c_m^F(-1,m-k) ((m-k)^2+4m)^{23/4}I_{23/2}\left(\sqrt{3}{\pi}\sqrt{((m-k)^2+4m)} \right)}
	{ c_m^F(-1,m)(m^2+4m)^{23/4}I_{23/2}\left(\sqrt{3}{\pi}\sqrt{(m^2+4m)}\right)}\, ,
\end{equation}
for $m\geq 7$.

From Lemma \ref{prop:polar-nminus1}, we can use the inequality \eqref{eq:bound-nminusone-polar} to show that each summand in \eqref{eq:nminusone-othercases} is smaller than
\begin{equation} \label{eq:bound-nminusone-remaining}
   \frac{2 (\kappa_4+\kappa_5\,{k^{1/2}})d(k-1) ((m-k)^2+4m)^{23/4}I_{23/2}\left(\sqrt{3}{\pi}\sqrt{((m-k)^2+4m)} \right)}
	{ (m^2+4m)^{23/4}I_{23/2}\left(\sqrt{3}{\pi}\sqrt{(m^2+4m)}\right)}\,.
\end{equation}
In the limit $m\to\infty$ the quantity in \eqref{eq:bound-nminusone-remaining} becomes
\begin{equation} \label{eq:limit-minfinity-nminusone}
    2 (\kappa_4+\kappa_5\,{k^{1/2}})d(k-1) e^{-\sqrt{3}\pi k} \,.
\end{equation}
To bound the summands in \eqref{eq:nminusone-othercases} for all possible values of $m$ one can proceed as in the previous cases. In the range $m \geq k$ the quantity \eqref{eq:bound-nminusone-remaining} has only one minimum, being bounded by its values at the extrema of the domain. The limit $m\to\infty$ is given by \eqref{eq:limit-minfinity-nminusone} and the first possible value $m$ can take for each $k$ is $m=k$, i.e.,
\begin{equation}
   \frac{2 (\kappa_4+\kappa_5 \,{k^{1/2}})d(k-1) (4k)^{23/4}I_{23/2}\left(\sqrt{3}{\pi}\sqrt{4k)} \right)}
	{ (k^2+4k)^{23/4}I_{23/2}\left(\sqrt{3}{\pi}\sqrt{(k^2+4k)}\right)}\,.
\end{equation}
Using the bound \eqref{eq:ratio-bessels-bound} and the second inequality in \eqref{eq:dedekindeta-bounds}, we can write
\begin{equation}
\begin{split}
      & \frac{2 (\kappa_4+\kappa_5 \,{k^{1/2}})d(k-1) (4k)^{23/4}I_{23/2}\left(\sqrt{3}{\pi}\sqrt{4k)} \right)}
	{ (k^2+4k)^{23/4}I_{23/2}\left(\sqrt{3}{\pi}\sqrt{(k^2+4k)}\right)} < \\
  & \left(1+\frac{1}{12} \right)\frac{4\pi  (\kappa_4+ \kappa_5 \,{k^{1/2}}) (4k)^{11/2}}
	{(k-1)^{13/2} (k^2+4k)^{11/2}} \, I_{13}\left(4\pi \sqrt{k-1} \right) \,e^{-\sqrt{3}\pi(\sqrt{k^2+4k}-2\sqrt{k})} = e_f(k)\,.
\end{split}
\end{equation}
This final estimate $e_f(k)$  is larger than \eqref{eq:limit-minfinity-nminusone} for all $k\geq 7$, and decreasing with $k$. Therefore we can bound \eqref{eq:nminusone-othercases} by
\begin{equation}
     \sum_{k=7}^{\infty}e_f(k) \,<\, 
  \int_{6}^{\infty} dk \, e_f(k) \,<\, 0.005 \,,
\end{equation}
so that
\begin{equation}
    \sum_{k=1}^m\frac{2\,c_m^F(-1,m-k) ((m-k))^2+4m)^{23/4}I_{23/2}\left(\frac{\pi}{m}\sqrt{\Delta((m-k))^2+4m)} \right)}
	{ c_m^F(-1,m)(m^2+4m)^{23/4}I_{23/2}\left(\frac{\pi}{m}\sqrt{\Delta(m^2+4m)}\right)} < 0.3 \,.
\end{equation}
\end{proof}
\subsection{Final bounds}
The proof of Theorem \ref{Main-Theorem} for the cases $m\geq 6$ follows from combining the results of the previous lemmas, as follows

\begin{proof}[Proof of Theorem \ref{Main-Theorem} for $m\geq 6$]
    From the Rademacher expansion for the coefficients $c_m^F({n},{\ell})$ \eqref{eq:Rademacher-formula} with $\Delta >0$, and $n\geq m$, $0\leq \ell \leq m$, we can bound $ (-1)^{\ell+1}c_m^F({n},{\ell}) $ from below by the positive contribution minus the absolute value of the other contributions,
    \begin{equation}
    \begin{split}
              (-1)^{\ell+1}c_m^F({n},{\ell})  \, \geq \, & \sum_{\substack{\tilde{n}\geq -1 \\ \tilde{\Delta}<0} } C_{\mathrm{pos.}}(\Delta, m, \tilde{n})\left( 1-\frac{2}{21} \right) \\
               & - 2\pi\left(1+\frac{2}{21} \right) \sum_{\substack{0\leq \tilde{\ell} < m \\ \tilde{n}\geq -1,\\
		\tilde{\Delta}<0}}  \frac{2}{\sqrt{2m}}\,c^F_m(\tilde{n},\tilde{\ell})
	  \left(\frac{ \vert\tilde{\Delta} \vert}{\Delta} \right)^{23/4} I_{23/2}\left(\frac{\pi}{ m}\sqrt{\Delta\vert\tilde{\Delta}\vert} \right) \\
              & - \left\vert \sum\limits_{\gamma=1}^{+\infty} C_{12}(m,n,\ell,\gamma) \right\vert
              -\left\vert \sum\limits_{\gamma=1}^{+\infty} C_{25/2}(m,n,\ell,\gamma) \right\vert \,,
    \end{split}
    \end{equation}
    where we have used Lemma \ref{lemma:bounds-besselterms-general} for writing the bounds on the absolute values in the first and second line. The positive contribution with $\tilde{n}=-1$ is the dominant term which was used in Lemma \ref{prop:I12-I25-total-bound} to bound the last line. We  have
    \begin{equation}
    \begin{split}
               & C_{\mathrm{pos.}}(\Delta, m, -1)\left( 1-\frac{2}{21} \right) - \left\vert \sum\limits_{\gamma=1}^{+\infty} C_{12}(m,n,\ell,\gamma) \right\vert
              -\left\vert \sum\limits_{\gamma=1}^{+\infty} C_{25/2}(m,n,\ell,\gamma) \right\vert \\
        \= \,      &  C_{\mathrm{pos.}}(\Delta, m, -1)\left( 1- e_{12}(\Delta,m) - e_{25/2}(\Delta,m)-  \frac{2}{21} \right)\,.
    \end{split}
    \end{equation}
    Finally, from Lemmas \ref{prop:I12-I25-total-bound}, \ref{lemma:total-contr-generaln}, \ref{lemma:total-contr-n0} and \ref{lemma:total-contr-nminus1} we have
    \begin{equation}
    \begin{split}
              (-1)^{\ell+1}c_m^F({n},{\ell})  \, \geq \, &  C_{\mathrm{pos.}}(\Delta, m, -1)  \left( 1 - e_{23/2}(\Delta,m,-1)\left(1+\frac{2}{21} \right)\right. \\
               &     - e_{12}(\Delta,m) - e_{25/2}(\Delta,m)-  \frac{2}{21} \Bigg)
                \\
              &  +C_{\mathrm{pos.}}(\Delta, m, 0)  \left( 1 -e_{23/2}(\Delta,m,0)\left(1+\frac{2}{21} \right) -  \frac{2}{21} \right)  \\
               & +\sum_{\substack{\tilde{n}\geq 1 \\ \tilde{\Delta}<0} } C_{\mathrm{pos.}}(\Delta, m, \tilde{n}) \left( 1 - e_{23/2}(\Delta,m,\tilde{n})\left(1+\frac{2}{21} \right) -\frac{2}{21} \right) \, > 0 \,,
    \end{split}
    \end{equation}
    since for $m\geq 6$ and $\Delta \geq 3m^2$,
    \begin{equation}
        e_{23/2}(\Delta,m,-1)\left(1+\frac{2}{21} \right)  + e_{12}(\Delta,m) + e_{25/2}(\Delta,m) + \frac{2}{21} \,<\, 0.45 \,<\,1\,,
    \end{equation}
    \begin{equation}
         e_{23/2}(\Delta,m,0)\left(1+\frac{2}{21} \right) +  \frac{2}{21}  \,<\, 0.25 \,<\,1\,,
    \end{equation}
    and 
    \begin{equation}
       e_{23/2}(\Delta,m,\tilde{n})\left(1+\frac{2}{21} \right) +\frac{2}{21} \,<\, 0.35 \,<\,1\,.
    \end{equation}
\end{proof}
The cases $m=1,2$ were proven in \cite{Bringmann:2012zr} and the proof for the remaining cases $m=3,4,5$ is given in Appendix \ref{appendixA}.

\section*{Acknowledgements}

I would like to thank the participants of the GASP workshop at IPMU Tokyo, Justin David, and Sameer Murthy for useful and enjoyable discussions. I would like to thank IPMU Tokyo and the University of Amsterdam for hospitality, and the Isaac Newton Institute for Mathematical Sciences, Cambridge, for support and hospitality during the programme Black holes: bridges between number theory and holographic quantum information, where work on this paper was undertaken. This work was supported by EPSRC grant EP/R014604/1.

\newpage

\appendix

\section{The first three cases \texorpdfstring{$m=3,4,5$}{m = 3, 4, 5}}
\label{appendixA}

In this Appendix we prove the statement of the main Theorem \ref{Main-Theorem} for the cases $m=3,4,5$. The proof is illustrative of the power of the Rademacher expansion to easily prove the positivity on a case by case basis when the polar coefficients are known. %

\begin{proof}[Proof of Theorem \ref{Main-Theorem} for  $m =3,4,5$]
    Starting with $m=3$, we first compute the polar coefficients $c^F_3(\tilde{n},\tilde{\ell})$ with $\tilde{\Delta}<0$ that enter the formula \eqref{eq:Rademacher-formula}. As earlier, since we are instructed to sum over $\tilde{\ell}\in\mathbb{Z}/2m\mathbb{Z}$, we can choose the set of representatives $-m < \tilde{\ell}\leq m$, which we can further reduce to $0\leq\tilde{\ell}\leq m$ due to the symmetry $c^F_m(\tilde{n},\tilde{\ell})=c^F_m(\tilde{n},-\tilde{\ell})$. We have
\begin{center}
\begin{tabular}{c|cccc}
	 $c^F_3(\tilde{n},\tilde{\ell})$& $\tilde{\ell} = 3$ & $\tilde{\ell} = 2$ & $\tilde{\ell}=1$ & $\tilde{\ell} = 0$  \\
	 \hline 
	 $\tilde{n}=-1$ & 5 & 96  & 972 & 6404 \\
	 $\tilde{n}= 0$ & 1728 & 15600 &  85176 & \\
\end{tabular}
\end{center}
where these coefficients are computed from the formula \eqref{polar-coef-formula}. More explicitly,
\begin{equation}
\begin{split}
		c_3^F(-1,3)_{\tilde{\Delta}=-21} &\= 5d(-1)d(-1) \= 5, \\
		 c_3^F(-1,2)_{\tilde{\Delta}=-16} &\= 4d(-1)d(0) \= 96, \\
		 c_3^F(-1,1)_{\tilde{\Delta}=-13} &\= 3d(-1)d(1) \= 972, \\
		 c_3^F(-1,0)_{\tilde{\Delta}=-12} &\= 2d(-1)d(2)+4d(-1)d(-1) \= 6404,
\end{split}
\end{equation}
and
\begin{equation}
\begin{split}
		c_3^F(0,3)_{\tilde{\Delta}=-9} & \= 3d(0)d(0) \= 1728,\\
		 c_3^F(0,2)_{\tilde{\Delta}=-4} & \= 2d(0)d(1)+2d(0)d(-1) \= 15600,\\
		  c_3^F(0,1)_{\tilde{\Delta}=-1} & \= d(0)d(2)+d(0)d(1)+d(0)d(0)+d(0)d(-1) \= 85176.
\end{split}
\end{equation}
From these values one can directly see that, roughly, for smaller $\vert\tilde{\Delta}\vert$ the polar coefficients are larger and therefore could compete with the Bessel function growth. 
As in the main text, we want to show now that, for each $\tilde{n}$, the contribution with $\tilde{\ell}=3$ to the Rademacher expansion is larger than the contributions with $\tilde{\ell}<3$, being enough to prove it for $\gamma=1$. 
That is, for a given $m$, we want to show that for each $\tilde{n}$, the term
\begin{equation}
	 2\pi \,  c^F_m(\tilde{n},m)\frac{1}{\sqrt{2m}}
	  \left(\frac{ \vert\tilde{\Delta} \vert}{\Delta} \right)^{23/4} I_{23/2}\left(\frac{\pi}{ m}\sqrt{\Delta\vert\tilde{\Delta}\vert} \right)
\end{equation} 
majorates
\begin{equation}
	\left\vert  (-1)^{\ell+1}  2\pi \sum_{\substack{ -m < \tilde{\ell}<m\\
		\tilde{\Delta}<0}} c^F_m(\tilde{n},\tilde{\ell}){{\rm Kl}(\frac{\Delta}{4m},\frac{\tilde{\Delta}}{4m},1,\psi)_{\ell \tilde{\ell}}}
	  \left(\frac{ \vert\tilde{\Delta} \vert}{\Delta} \right)^{23/4} I_{23/2}\left(\frac{\pi}{ m}\sqrt{\Delta\vert\tilde{\Delta}\vert} \right) \right \vert \,.
\end{equation} 
For $\gamma=1$,
\begin{equation} \label{eq:kloosterman-gamma1-bound}
    \left\vert  {{\rm Kl}(\frac{\Delta}{4m},\frac{\tilde{\Delta}}{4m},1,\psi)_{\ell \tilde{\ell}}}
	   \right \vert \,\leq\,
    \frac{1}{\sqrt{2m}}\,,
\end{equation}
so that the ratio to bound is
\begin{equation}
   \sum_{\substack{ -m < \tilde{\ell}<m\\
		\tilde{\Delta}<0}} \frac{2  c^F_m(\tilde{n},\tilde{\ell})
	  { \vert 4 m \tilde{n}-\tilde{\ell}^2 \vert} ^{23/4} I_{23/2}\left(\frac{\pi}{ m}\sqrt{\Delta\vert 4 m \tilde{n}-\tilde{\ell}^2\vert} \right)}{  c^F_m(\tilde{n},m)
	  { \vert 4m \tilde{n}-m^2 \vert} ^{23/4} I_{23/2}\left(\frac{\pi}{ m}\sqrt{\Delta\vert4m \tilde{n}-m^2\vert} \right)}\,.
\end{equation}
For the case $m=3$ we can just evaluate each numerical value. Using the fact that the function
\begin{equation}
	\frac{I_\nu(ax)}{I_\nu(bx)}
\end{equation}
is a decreasing function of $x$ for $0<a<b$ and $\nu>0$, we can use the inequality $\Delta =4mn-\ell^2\geq 3m^2 = 27$ to bound
\begin{equation} \label{eq:m3-bound-nminus1}
\begin{split}
	e_{23/2}(\Delta,3,-1 ) \=  &\sum_{\tilde{\ell}=0}^2 \frac{2c_3^F(-1,\tilde{\ell})}{c_3^F(-1,3)} \frac{\vert 12+\tilde{\ell}^2\vert^{23/4}}{\vert 12+9^2\vert^{23/4}}\frac{I_{23/2}\left( \frac{\pi}{3} \sqrt{\Delta(12+\tilde{\ell}^2)} \right)}{I_{23/2}\left( \frac{\pi}{3} \sqrt{\Delta(12+9^2)} \right)} \,\leq\, \\
	& 	\sum_{\tilde{\ell}=0}^2 \frac{2c_3^F(-1,\tilde{\ell})}{c_3^F(-1,3)} \frac{\vert 12+\tilde{\ell}^2\vert^{23/4}}{\vert 12+9^2\vert^{23/4}}\frac{I_{23/2}\left( \sqrt{3}{\pi}\sqrt{(12+\tilde{\ell}^2)} \right)}{I_{23/2}\left( \sqrt{3}{\pi} \sqrt{(12+9^2)} \right)} \,<\,0.44	 \,.
\end{split}
\end{equation}
For $\tilde{n}=0$ applying the same reasoning we get
\begin{equation} \label{eq:m3-bound-n0}
	e_{23/2}(\Delta,3,0 ) \= \sum_{\tilde{\ell}=1}^2 \frac{2c_3^F(0,\tilde{\ell})}{c_3^F(0,3)} \frac{\vert \tilde{\ell}^2\vert^{23/4}}{\vert 9^2\vert^{23/4}}\frac{I_{23/2}\left( \frac{\pi}{3} \sqrt{\Delta(\tilde{\ell}^2)} \right)}{I_{23/2}\left( \frac{\pi}{3} \sqrt{\Delta(9^2)} \right)} <0.0002\,.
\end{equation}

For the case $m=4$ we follow exactly the same reasoning as previously, being therefore less explicit. The non-zero polar coefficients $c^F_4(\tilde{n},\tilde{\ell})$ are
\begin{center}
\begin{tabular}{c|ccccc}
	 $c^F_4(\tilde{n},\tilde{\ell})$& $\tilde{\ell} = 4$& $\tilde{\ell} = 3$ & $\tilde{\ell} = 2$ & $\tilde{\ell}=1$ & $\tilde{\ell} = 0$  \\
	 \hline 
	 $\tilde{n}=-1$ & 6 & 120  & 1296 & 9600 & 51396 \\
	 $\tilde{n}= 0$ & 2304 & 23328 &  154752 & 700776 \\
\end{tabular}
\end{center}
With this information we can compute
\begin{equation} \label{eq:m4-bound-nminus1}
	e_{23/2}(\Delta,4,-1 )\= \sum_{\tilde{\ell}=0}^3 \frac{2c_4^F(-1,\tilde{\ell})}{c_4^F(-1,4)} \frac{\vert 16+\tilde{\ell}^2\vert^{23/4}}{\vert 16+16^2\vert^{23/4}}\frac{I_{23/2}\left( \frac{\pi}{3} \sqrt{\Delta(16+\tilde{\ell}^2)} \right)}{I_{23/2}\left( \frac{\pi}{3} \sqrt{\Delta(16+16^2)} \right)} \,<\,  0.28 \,,
\end{equation}
and
\begin{equation} \label{eq:m4-bound-n0}
	 e_{23/2}(\Delta,4,0 )\=\sum_{\tilde{\ell}=1}^3 \frac{2c_4^F(0,\tilde{\ell})}{c_4^F(0,4)} \frac{\vert \tilde{\ell}^2\vert^{23/4}}{\vert 16^2\vert^{23/4}}\frac{I_{23/2}\left( \frac{\pi}{3} \sqrt{\Delta(\tilde{\ell}^2)} \right)}{I_{23/2}\left( \frac{\pi}{3} \sqrt{\Delta(16^2)} \right)} < 0.002\,,
\end{equation}
where we have used again $\Delta \geq 3m^2$. 

Finally, the polar coefficients for $m=5$ are
\begin{center}
\begin{tabular}{c|cccccc}
	 $c^F_5(\tilde{n},\tilde{\ell})$& $\tilde{\ell} = 5$& $\tilde{\ell} = 4$& $\tilde{\ell} = 3$ & $\tilde{\ell} = 2$ & $\tilde{\ell}=1$ & $\tilde{\ell} = 0$  \\
	 \hline 
	 $\tilde{n}=-1$ & 7 & 144  & 1620 & 12800 & 76955 & 353808 \\
	 $\tilde{n}= 0$ & 2880 & 31104 &  230472 & 1246800 & 4930920 \\
  $\tilde{n}= 1$ & 315255  \\
\end{tabular}
\end{center}
which yield
\begin{equation} \label{eq:m5-bound-nminus1}
	e_{23/2}(\Delta,5,-1 )\= \sum_{\tilde{\ell}=0}^4 \frac{2c_5^F(-1,\tilde{\ell})}{c_5^F(-1,5)} \frac{\vert 20+\tilde{\ell}^2\vert^{23/4}}{\vert 16+16^2\vert^{23/4}}\frac{I_{23/2}\left( \frac{\pi}{3} \sqrt{\Delta(20+\tilde{\ell}^2)} \right)}{I_{23/2}\left( \frac{\pi}{3} \sqrt{\Delta(20+25^2)} \right)} \,<\,  0.24 \,,
\end{equation}
and
\begin{equation} \label{eq:m5-bound-n0}
	e_{23/2}(\Delta,5,0 )\=\sum_{\tilde{\ell}=1}^4 \frac{2c_5^F(0,\tilde{\ell})}{c_5^F(0,5)} \frac{\vert \tilde{\ell}^2\vert^{23/4}}{\vert 25^2\vert^{23/4}}\frac{I_{23/2}\left( \frac{\pi}{3} \sqrt{\Delta(\tilde{\ell}^2)} \right)}{I_{23/2}\left( \frac{\pi}{3} \sqrt{\Delta(25^2)} \right)} \,<\, 0.005\,.
\end{equation}

To complete the proof of positivity of these Fourier coefficients we need to also bound the contribution coming from the mock terms with Bessel functions of different index. As in the cases $m\geq 6$, we bound them with respect to the largest contribution coming from the $\tilde{n}=-1$ term. For the $I_{12}$ terms, from \eqref{eq:bound-I12-highergamma}, and following the same reasoning as in Lemma \ref{prop:I12-I25-total-bound}, the quantity
\begin{equation} \label{eq:threecases-i12-bounding}
  e_{12}(\Delta,m) \= \frac{\left\vert \sum\limits_{\gamma=1}^{+\infty} C_{12}(m,n,\ell,\gamma) \right\vert}{2\pi \, c^F_m(-1,m)\frac{1}{\sqrt{2m}}
	  \left(\frac{ \vert\tilde{\Delta} \vert}{\Delta} \right)^{23/4} I_{23/2}\left(\frac{\pi}{ m}\sqrt{\Delta\vert\tilde{\Delta}\vert} \right)}
\end{equation} 
can be shown to decrease with $m$, with the values for $m=3,4,5$ satisfying $e_{12}(\Delta,3) < 0.15$, $e_{12}(\Delta,4)<0.03$, and $e_{12}(\Delta,5) < 0.004$.

The contribution from the terms containing a Bessel function of index $25/2$ is treated differently that for the cases $m\geq 6$, since a tighter bound is needed. We first focus on the $\gamma=1$ contribution. Its absolute value, using \eqref{eq:kloosterman-gamma1-bound}, can be written as
\begin{equation} \label{eq:C25/2-m345-gamma1}\begin{split}
       & \left\vert C_{25/2}(m,n,\ell,1)\right\vert \= 
             \frac{1}{2\pi} d(m) \sum_{\tilde{\ell}\in\mathbb{Z}/2m\mathbb{Z}} \frac{1}{\sqrt{2m}} \left(\frac{4m}{\Delta} \right)^{25/4} 
	  \\
&   \int_{\frac{-1}{\sqrt{m}}}^{\frac{1}{\sqrt{m}}}dx'  \left\vert f_{1,\tilde{\ell},m}(x')\right\vert  (1-mx'^2)^{25/4}
	I_{25/2}\left(\frac{2\pi}{\sqrt{m}}\sqrt{\Delta(1-mx'^2)}\right).   \\
\end{split}
\end{equation}
The ratio of this quantity against the positive contribution \eqref{eq:def-positive-contr} with $\tilde{n}=-1$ can be majorated by
\begin{equation} \label{eq:C25/2-m345-gamma1-ratio}
        \begin{split}
            & \frac{1}{C_{\mathrm{pos}}(3m^2,m,-1)}\frac{1}{2\pi} d(m) \sum_{\tilde{\ell}\in\mathbb{Z}/2m\mathbb{Z}} \frac{1}{\sqrt{2m}} \left(\frac{4m}{3m^2} \right)^{25/4} 
	  \\
&   \int_{\frac{-1}{\sqrt{m}}}^{\frac{1}{\sqrt{m}}}dx'  \left\vert f_{1,\tilde{\ell},m}(x')\right\vert  (1-mx'^2)^{25/4}
	I_{25/2}\left({2\sqrt{3}\pi}\sqrt{m(1-mx'^2)}\right).   \\
\end{split}
\end{equation}
We can compute this quantity for $m=3,4,5$, yielding values which are smaller than $0.06, 0.01$ and $0.001$, respectively. For the higher gamma contributions we use the previous bounds \eqref{eq:gen-kloosterman-bound} and \eqref{eq:bound-f-function} together with the bounds for the sum over $\gamma$ as in Proposition \ref{prop:eta-bound} to obtain the inequality
\begin{equation} \label{eq:C25/2-m345-higher-gamma}
        \begin{split}
            &\left\vert \frac{1}{2\pi} d(m) \sum_{\gamma=2}^{+\infty}\sum_{\tilde{\ell}\in\mathbb{Z}/2m\mathbb{Z}} \sum_{\substack{ g \in \mathbb{Z}/2m\gamma \mathbb{Z} \\ g = \tilde{\ell} \text{ mod }2m}}\frac{{\rm Kl }( \frac{\Delta}{4m}, -1-\frac{g^2}{4m} ;\gamma,\psi)_{\ell\tilde{\ell}}}{\gamma^2} \left(\frac{4m}{\Delta} \right)^{25/4} \right.
	  \\
&  \left.  \int_{\frac{-1}{\sqrt{m}}}^{\frac{1}{\sqrt{m}}}dx'  f_{\gamma,g,m}(x')  (1-mx'^2)^{25/4}
	I_{25/2}\left(\frac{2\pi}{\gamma\sqrt{m}}\sqrt{\Delta(1-mx'^2)}\right) \right\vert  \\
 & \,\leq\, \frac{1}{2\pi} d(m) \pi^2  \left(m^{2} \frac{2}{21}\frac{\pi^2}{3}+\frac{2}{25}\frac{1}{3}\right) \left(\frac{4m}{\Delta} \right)^{25/4} \\
 & \int_{\frac{-1}{\sqrt{m}}}^{\frac{1}{\sqrt{m}}}dx'    (1-mx'^2)^{25/4}
	  I_{25/2}\left(\frac{2\pi}{\sqrt{m}}\sqrt{\Delta(1-mx'^2)}\right).
\end{split}
\end{equation}
The ratio of this quantity over the positive contribution \eqref{eq:def-positive-contr} can be bounded using $\Delta \geq 3m^2$ value, and the ratios for $\Delta = 3m^2$ have values smaller than $0.04, 0.008, 0.002$ for $m=3,4,5$, respectively, so that we have shown that
\begin{equation} \label{threecases-i25-bounding}
    e_{25/2}(\Delta,m) \= \frac{\left\vert \sum\limits_{\gamma=1}^{+\infty} C_{25/2}(m,n,\ell,\gamma) \right\vert}{2\pi \, c^F_m(-1,m)\frac{1}{\sqrt{2m}}
	  \left(\frac{ \vert\tilde{\Delta} \vert}{\Delta} \right)^{23/4} I_{23/2}\left(\frac{\pi}{ m}\sqrt{\Delta\vert\tilde{\Delta}\vert} \right)}\, ,
\end{equation} 
for $m=3,4,5$, satisfies $ e_{25/2}(\Delta,3) < 0.1$, $e_{25/2}(\Delta,4) < 0.02$ and $e_{25/2}(\Delta,5) <  0.003$.

We conclude by following the same steps as in the proof of Theorem \ref{Main-Theorem} for $m\geq 6$. We have
\begin{equation} 
\begin{split}
        (-1)^{\ell+1}c_m^F({n},{\ell})  \, \geq \, & \sum_{\substack{\tilde{n}\geq -1 \\ \tilde{\Delta}<0} } C_{\mathrm{pos.}}(\Delta, m, \tilde{n})\left( 1-\frac{2}{21} \right) \\
               & - 2\pi\left(1+\frac{2}{21} \right) \sum_{\substack{0\leq \tilde{\ell} < m \\ \tilde{n}\geq -1,\\
		\tilde{\Delta}<0}}  2c^F_m(\tilde{n},\tilde{\ell})
	  \left(\frac{ \vert\tilde{\Delta} \vert}{\Delta} \right)^{23/4} I_{23/2}\left(\frac{\pi}{ m}\sqrt{\Delta\vert\tilde{\Delta}\vert} \right) \\
              & - \left\vert \sum\limits_{\gamma=1}^{+\infty} C_{12}(m,n,\ell,\gamma) \right\vert
              -\left\vert \sum\limits_{\gamma=1}^{+\infty} C_{25/2}(m,n,\ell,\gamma) \right\vert \,,
\end{split}
\end{equation}
which we can write as follows for the three cases at hand
\begin{equation}
\begin{split}
        (-1)^{\ell+1}c_m^F({n},{\ell})  \, \geq \,  & C_{\mathrm{pos.}}(\Delta, m, -1) \left(1 -e_{23/2}(\Delta,m,-1 )\left(1+\frac{2}{21} \right) \right. \\
         & -e_{12}(\Delta,m) -e_{25/2}(\Delta,m) - \frac{2}{21} \Bigg) \\
       & +  C_{\mathrm{pos.}}(\Delta, m, 0)\left(1 -e_{23/2}(\Delta,m, 0)\left(1+\frac{2}{21} \right)- \frac{2}{21} \right) \, > \, 0 \,,
\end{split}
\end{equation}
where the last inequality follows from the previous numerical values for the bounds.

\end{proof}

\newpage

\bibliography{n=4}

\providecommand{\href}[2]{#2}\begingroup\raggedright\begin{thebibliography}{10}

\bibitem{Sen:2011ktd}
A.~Sen, ``{How Do Black Holes Predict the Sign of the Fourier Coefficients of Siegel Modular Forms?},'' {\em Gen. Rel. Grav.} {\bf 43} (2011) 2171--2183, \href{http://www.arXiv.org/abs/1008.4209}{{\tt 1008.4209}}.

\bibitem{Sen:1995in}
A.~Sen, ``{Extremal black holes and elementary string states},'' {\em Mod. Phys. Lett. A} {\bf 10} (1995) 2081--2094, \href{http://www.arXiv.org/abs/hep-th/9504147}{{\tt hep-th/9504147}}.

\bibitem{Strominger:1996sh}
A.~Strominger and C.~Vafa, ``{Microscopic origin of the Bekenstein-Hawking entropy},'' {\em Phys. Lett. B} {\bf 379} (1996) 99--104, \href{http://www.arXiv.org/abs/hep-th/9601029}{{\tt hep-th/9601029}}.

\bibitem{Sen:2008vm}
A.~Sen, ``{Quantum Entropy Function from AdS(2)/CFT(1) Correspondence},'' {\em Int. J. Mod. Phys. A} {\bf 24} (2009) 4225--4244, \href{http://www.arXiv.org/abs/0809.3304}{{\tt 0809.3304}}.

\bibitem{Banerjee:2009af}
N.~Banerjee, S.~Banerjee, R.~K. Gupta, I.~Mandal, and A.~Sen, ``{Supersymmetry, Localization and Quantum Entropy Function},'' {\em JHEP} {\bf 02} (2010) 091, \href{http://www.arXiv.org/abs/0905.2686}{{\tt 0905.2686}}.

\bibitem{Dabholkar:2010uh}
A.~Dabholkar, J.~Gomes, and S.~Murthy, ``{Quantum black holes, localization and the topological string},'' {\em JHEP} {\bf 06} (2011) 019, \href{http://www.arXiv.org/abs/1012.0265}{{\tt 1012.0265}}.

\bibitem{Dabholkar:2011ec}
A.~Dabholkar, J.~Gomes, and S.~Murthy, ``{Localization \& Exact Holography},'' {\em JHEP} {\bf 04} (2013) 062, \href{http://www.arXiv.org/abs/1111.1161}{{\tt 1111.1161}}.

\bibitem{Gupta:2012cy}
R.~K. Gupta and S.~Murthy, ``{All solutions of the localization equations for N=2 quantum black hole entropy},'' {\em JHEP} {\bf 02} (2013) 141, \href{http://www.arXiv.org/abs/1208.6221}{{\tt 1208.6221}}.

\bibitem{Dabholkar:2014ema}
A.~Dabholkar, J.~Gomes, and S.~Murthy, ``{Nonperturbative black hole entropy and Kloosterman sums},'' {\em JHEP} {\bf 03} (2015) 074, \href{http://www.arXiv.org/abs/1404.0033}{{\tt 1404.0033}}.

\bibitem{Gupta:2015gga}
R.~K. Gupta, Y.~Ito, and I.~Jeon, ``{Supersymmetric Localization for BPS Black Hole Entropy: 1-loop Partition Function from Vector Multiplets},'' {\em JHEP} {\bf 11} (2015) 197, \href{http://www.arXiv.org/abs/1504.01700}{{\tt 1504.01700}}.

\bibitem{Murthy:2015yfa}
S.~Murthy and V.~Reys, ``{Functional determinants, index theorems, and exact quantum black hole entropy},'' {\em JHEP} {\bf 12} (2015) 028, \href{http://www.arXiv.org/abs/1504.01400}{{\tt 1504.01400}}.

\bibitem{Murthy:2015zzy}
S.~Murthy and V.~Reys, ``{Single-centered black hole microstate degeneracies from instantons in supergravity},'' {\em JHEP} {\bf 04} (2016) 052, \href{http://www.arXiv.org/abs/1512.01553}{{\tt 1512.01553}}.

\bibitem{Jeon:2018kec}
I.~Jeon and S.~Murthy, ``{Twisting and localization in supergravity: equivariant cohomology of BPS black holes},'' {\em JHEP} {\bf 03} (2019) 140, \href{http://www.arXiv.org/abs/1806.04479}{{\tt 1806.04479}}.

\bibitem{deWit:2018dix}
B.~de~Wit, S.~Murthy, and V.~Reys, ``{BRST quantization and equivariant cohomology: localization with asymptotic boundaries},'' {\em JHEP} {\bf 09} (2018) 084, \href{http://www.arXiv.org/abs/1806.03690}{{\tt 1806.03690}}.

\bibitem{Iliesiu:2022kny}
L.~V. Iliesiu, S.~Murthy, and G.~J. Turiaci, ``{Black hole microstate counting from the gravitational path integral},'' \href{http://www.arXiv.org/abs/2209.13602}{{\tt 2209.13602}}.

\bibitem{LopesCardoso:2022hvc}
G.~Lopes~Cardoso, A.~Kidambi, S.~Nampuri, V.~Reys, and M.~Rossell\'o, ``{The Gravitational Path Integral for $ N=4$ BPS Black Holes from Black Hole Microstate Counting},'' {\em Annales Henri Poincare} {\bf 24} (2023), no.~10, 3305--3346, \href{http://www.arXiv.org/abs/2211.06873}{{\tt 2211.06873}}.

\bibitem{Sen:2023dps}
A.~Sen, ``{Revisiting localization for BPS black hole entropy},'' \href{http://www.arXiv.org/abs/2302.13490}{{\tt 2302.13490}}.

\bibitem{GonzalezLezcano:2023uar}
A.~Gonz\'alez~Lezcano, I.~Jeon, and A.~Ray, ``{Supersymmetry and complexified spectrum on Euclidean AdS2},'' {\em Phys. Rev. D} {\bf 108} (2023), no.~4, 045018, \href{http://www.arXiv.org/abs/2305.12925}{{\tt 2305.12925}}.

\bibitem{Dabholkar:2010rm}
A.~Dabholkar, J.~Gomes, S.~Murthy, and A.~Sen, ``{Supersymmetric Index from Black Hole Entropy},'' {\em JHEP} {\bf 04} (2011) 034, \href{http://www.arXiv.org/abs/1009.3226}{{\tt 1009.3226}}.

\bibitem{Banerjee:2009uk}
N.~Banerjee, I.~Mandal, and A.~Sen, ``{Black Hole Hair Removal},'' {\em JHEP} {\bf 07} (2009) 091, \href{http://www.arXiv.org/abs/0901.0359}{{\tt 0901.0359}}.

\bibitem{Jatkar:2009yd}
D.~P. Jatkar, A.~Sen, and Y.~K. Srivastava, ``{Black Hole Hair Removal: Non-linear Analysis},'' {\em JHEP} {\bf 02} (2010) 038, \href{http://www.arXiv.org/abs/0907.0593}{{\tt 0907.0593}}.

\bibitem{Sen:2009vz}
A.~Sen, ``{Arithmetic of Quantum Entropy Function},'' {\em JHEP} {\bf 08} (2009) 068, \href{http://www.arXiv.org/abs/0903.1477}{{\tt 0903.1477}}.

\bibitem{Bachas:1996bp}
C.~Bachas and E.~Kiritsis, ``{F(4) terms in N=4 string vacua},'' {\em Nucl. Phys. B Proc. Suppl.} {\bf 55} (1997) 194--199, \href{http://www.arXiv.org/abs/hep-th/9611205}{{\tt hep-th/9611205}}.

\bibitem{Gregori:1997hi}
A.~Gregori, E.~Kiritsis, C.~Kounnas, N.~A. Obers, P.~M. Petropoulos, and B.~Pioline, ``{R**2 corrections and nonperturbative dualities of N=4 string ground states},'' {\em Nucl. Phys. B} {\bf 510} (1998) 423--476, \href{http://www.arXiv.org/abs/hep-th/9708062}{{\tt hep-th/9708062}}.

\bibitem{Dijkgraaf:1996it}
R.~Dijkgraaf, E.~P. Verlinde, and H.~L. Verlinde, ``{Counting dyons in N=4 string theory},'' {\em Nucl. Phys. B} {\bf 484} (1997) 543--561, \href{http://www.arXiv.org/abs/hep-th/9607026}{{\tt hep-th/9607026}}.

\bibitem{Chattopadhyaya:2020yzl}
A.~Chattopadhyaya and J.~R. David, ``{Horizon states and the sign of their index in ${\mathcal N}=4$ dyons},'' {\em JHEP} {\bf 03} (2021) 106, \href{http://www.arXiv.org/abs/2010.08967}{{\tt 2010.08967}}.

\bibitem{Dabholkar:2009dq}
A.~Dabholkar, M.~Guica, S.~Murthy, and S.~Nampuri, ``{No entropy enigmas for N=4 dyons},'' {\em JHEP} {\bf 06} (2010) 007, \href{http://www.arXiv.org/abs/0903.2481}{{\tt 0903.2481}}.

\bibitem{Cheng:2007ch}
M.~C.~N. Cheng and E.~Verlinde, ``{Dying Dyons Don't Count},'' {\em JHEP} {\bf 09} (2007) 070, \href{http://www.arXiv.org/abs/0706.2363}{{\tt 0706.2363}}.

\bibitem{Sen:2007qy}
A.~Sen, ``{Black Hole Entropy Function, Attractors and Precision Counting of Microstates},'' {\em Gen. Rel. Grav.} {\bf 40} (2008) 2249--2431, \href{http://www.arXiv.org/abs/0708.1270}{{\tt 0708.1270}}.

\bibitem{Bringmann:2012zr}
K.~Bringmann and S.~Murthy, ``{On the positivity of black hole degeneracies in string theory},'' {\em Commun. Num. Theor Phys.} {\bf 07} (2013) 15--56, \href{http://www.arXiv.org/abs/1208.3476}{{\tt 1208.3476}}.

\bibitem{Dabholkar:2012nd}
A.~Dabholkar, S.~Murthy, and D.~Zagier, ``{Quantum Black Holes, Wall Crossing, and Mock Modular Forms},'' \href{http://www.arXiv.org/abs/1208.4074}{{\tt 1208.4074}}.

\bibitem{Banerjee:2008ky}
N.~Banerjee, D.~P. Jatkar, and A.~Sen, ``{Asymptotic Expansion of the N=4 Dyon Degeneracy},'' {\em JHEP} {\bf 05} (2009) 121, \href{http://www.arXiv.org/abs/0810.3472}{{\tt 0810.3472}}.

\bibitem{Murthy:2009dq}
S.~Murthy and B.~Pioline, ``{A Farey tale for N=4 dyons},'' {\em JHEP} {\bf 09} (2009) 022, \href{http://www.arXiv.org/abs/0904.4253}{{\tt 0904.4253}}.

\bibitem{Bringmann:2010sd}
K.~Bringmann and J.~Manschot, ``{From sheaves on $P^2$ to a generalization of the Rademacher expansion},'' {\em Am. J. Math.} {\bf 135} (2013), no.~4, 1039--1065, \href{http://www.arXiv.org/abs/1006.0915}{{\tt 1006.0915}}.

\bibitem{Ferrari:2017msn}
F.~Ferrari and V.~Reys, ``{Mixed Rademacher and BPS Black Holes},'' {\em JHEP} {\bf 07} (2017) 094, \href{http://www.arXiv.org/abs/1702.02755}{{\tt 1702.02755}}.

\bibitem{Chowdhury:2019mnb}
A.~Chowdhury, A.~Kidambi, S.~Murthy, V.~Reys, and T.~Wrase, ``{Dyonic black hole degeneracies in $\mathcal{N} = 4$ string theory from Dabholkar-Harvey degeneracies},'' {\em JHEP} {\bf 10} (2020) 184, \href{http://www.arXiv.org/abs/1912.06562}{{\tt 1912.06562}}.

\bibitem{LopesCardoso:2020pmp}
G.~L. Cardoso, S.~Nampuri, and M.~Rossell\'o, ``{Arithmetic of decay walls through continued fractions: a new exact dyon counting solution in $ \mathcal{N} $ = 4 CHL models},'' {\em JHEP} {\bf 03} (2021) 154, \href{http://www.arXiv.org/abs/2007.10302}{{\tt 2007.10302}}.

\bibitem{Cardoso:2021gfg}
G.~L. Cardoso, S.~Nampuri, and M.~Rossell\'o, ``{Rademacher expansion of a Siegel modular form for $\mathcal {N} = 4$ counting},'' \href{http://www.arXiv.org/abs/2112.10023}{{\tt 2112.10023}}.

\bibitem{Chowdhury:2014yca}
A.~Chowdhury, R.~S. Garavuso, S.~Mondal, and A.~Sen, ``{BPS State Counting in N=8 Supersymmetric String Theory for Pure D-brane Configurations},'' {\em JHEP} {\bf 10} (2014) 186, \href{http://www.arXiv.org/abs/1405.0412}{{\tt 1405.0412}}.

\bibitem{Chowdhury:2015gbk}
A.~Chowdhury, R.~S. Garavuso, S.~Mondal, and A.~Sen, ``{Do All BPS Black Hole Microstates Carry Zero Angular Momentum?},'' {\em JHEP} {\bf 04} (2016) 082, \href{http://www.arXiv.org/abs/1511.06978}{{\tt 1511.06978}}.

\bibitem{Maji:2023ims}
S.~Maji and A.~Chowdhury, ``{Counting $\mathcal{N} = 8$ Black Holes as Algebraic Varieties},'' \href{http://www.arXiv.org/abs/2311.04786}{{\tt 2311.04786}}.

\bibitem{Chaudhuri:1995fk}
S.~Chaudhuri, G.~Hockney, and J.~D. Lykken, ``{Maximally supersymmetric string theories in D \ensuremath{<} 10},'' {\em Phys. Rev. Lett.} {\bf 75} (1995) 2264--2267, \href{http://www.arXiv.org/abs/hep-th/9505054}{{\tt hep-th/9505054}}.

\bibitem{Chattopadhyaya:2017ews}
A.~Chattopadhyaya and J.~R. David, ``{Dyon degeneracies from Mathieu moonshine symmetry},'' {\em Phys. Rev. D} {\bf 96} (2017), no.~8, 086020, \href{http://www.arXiv.org/abs/1704.00434}{{\tt 1704.00434}}.

\bibitem{Chattopadhyaya:2018xvg}
A.~Chattopadhyaya and J.~R. David, ``{Properties of dyons in $ \mathcal{N} $ = 4 theories at small charges},'' {\em JHEP} {\bf 05} (2019) 005, \href{http://www.arXiv.org/abs/1810.12060}{{\tt 1810.12060}}.

\bibitem{Govindarajan:2022nzb}
S.~Govindarajan, S.~Samanta, P.~Shanmugapriya, and A.~Virmani, ``{Positivity of discrete information for CHL black holes},'' {\em Nucl. Phys. B} {\bf 987} (2023) 116095, \href{http://www.arXiv.org/abs/2205.08726}{{\tt 2205.08726}}.

\bibitem{Chattopadhyaya:2023uov}
A.~Chattopadhyaya, ``{Sign of BPS index for ${\mathcal N}=4$ dyons},'' \href{http://www.arXiv.org/abs/2303.03913}{{\tt 2303.03913}}.

\bibitem{Harrison:2022zee}
S.~M. Harrison, J.~A. Harvey, and N.~M. Paquette, ``{Snowmass White Paper: Moonshine},'' \href{http://www.arXiv.org/abs/2201.13321}{{\tt 2201.13321}}.

\bibitem{Cheng:2008gx}
M.~C.~N. Cheng, {\em {The Spectra of Supersymmetric States in String Theory}}.
\newblock PhD thesis, Amsterdam U., 2008.
\newblock \href{http://www.arXiv.org/abs/0807.3099}{{\tt 0807.3099}}.

\bibitem{Mandal:2010cj}
I.~Mandal and A.~Sen, ``{Black Hole Microstate Counting and its Macroscopic Counterpart},'' {\em Class. Quant. Grav.} {\bf 27} (2010) 214003, \href{http://www.arXiv.org/abs/1008.3801}{{\tt 1008.3801}}.

\bibitem{Dabholkar:2012zz}
A.~Dabholkar and S.~Nampuri, ``{Quantum black holes},'' {\em Lect. Notes Phys.} {\bf 851} (2012) 165--232, \href{http://www.arXiv.org/abs/1208.4814}{{\tt 1208.4814}}.

\bibitem{Banerjee:2007sr}
S.~Banerjee and A.~Sen, ``{Duality orbits, dyon spectrum and gauge theory limit of heterotic string theory on T**6},'' {\em JHEP} {\bf 03} (2008) 022, \href{http://www.arXiv.org/abs/0712.0043}{{\tt 0712.0043}}.

\bibitem{Banerjee:2008ri}
S.~Banerjee and A.~Sen, ``{S-duality Action on Discrete T-duality Invariants},'' {\em JHEP} {\bf 04} (2008) 012, \href{http://www.arXiv.org/abs/0801.0149}{{\tt 0801.0149}}.

\bibitem{Banerjee:2008pu}
S.~Banerjee, A.~Sen, and Y.~K. Srivastava, ``{Partition Functions of Torsion \ensuremath{>} 1 Dyons in Heterotic String Theory on T**6},'' {\em JHEP} {\bf 05} (2008) 098, \href{http://www.arXiv.org/abs/0802.1556}{{\tt 0802.1556}}.

\bibitem{Dabholkar:2008zy}
A.~Dabholkar, J.~Gomes, and S.~Murthy, ``{Counting all dyons in N =4 string theory},'' {\em JHEP} {\bf 05} (2011) 059, \href{http://www.arXiv.org/abs/0803.2692}{{\tt 0803.2692}}.

\bibitem{Dabholkar:1989jt}
A.~Dabholkar and J.~A. Harvey, ``{Nonrenormalization of the Superstring Tension},'' {\em Phys. Rev. Lett.} {\bf 63} (1989) 478.

\bibitem{David:2006yn}
J.~R. David and A.~Sen, ``{CHL Dyons and Statistical Entropy Function from D1-D5 System},'' {\em JHEP} {\bf 11} (2006) 072, \href{http://www.arXiv.org/abs/hep-th/0605210}{{\tt hep-th/0605210}}.

\bibitem{Shih:2005uc}
D.~Shih, A.~Strominger, and X.~Yin, ``{Recounting Dyons in N=4 string theory},'' {\em JHEP} {\bf 10} (2006) 087, \href{http://www.arXiv.org/abs/hep-th/0505094}{{\tt hep-th/0505094}}.

\bibitem{Gaiotto:2005hc}
D.~Gaiotto, ``{Re-recounting dyons in N=4 string theory},'' \href{http://www.arXiv.org/abs/hep-th/0506249}{{\tt hep-th/0506249}}.

\bibitem{Bossard:2018rlt}
G.~Bossard, C.~Cosnier-Horeau, and B.~Pioline, ``{Exact effective interactions and 1/4-BPS dyons in heterotic CHL orbifolds},'' {\em SciPost Phys.} {\bf 7} (2019), no.~3, 028, \href{http://www.arXiv.org/abs/1806.03330}{{\tt 1806.03330}}.

\bibitem{eichler2013theory}
M.~Eichler and D.~Zagier, {\em The Theory of Jacobi Forms}.
\newblock Progress in Mathematics. Birkh{\"a}user Boston, 2013.

\bibitem{Gritsenko:1996kt}
V.~A. Gritsenko and V.~V. Nikulin, ``{Automorphic forms and Lorentzian Kac-Moody algebras. Part 1.},''.

\bibitem{Gritsenko:1996tm}
V.~A. Gritsenko and V.~V. Nikulin, ``{Automorphic forms and Lorentzian Kac-Moody algebras. Part 2},'' \href{http://www.arXiv.org/abs/alg-geom/9611028}{{\tt alg-geom/9611028}}.

\bibitem{Borcherds1995}
R.~E. Borcherds, ``Automorphic forms on os + 2,2(r) and infinte products.,'' {\em Inventiones mathematicae} {\bf 120} (1995), no.~1, 161--214.

\bibitem{Sen:2007vb}
A.~Sen, ``{Walls of Marginal Stability and Dyon Spectrum in N=4 Supersymmetric String Theories},'' {\em JHEP} {\bf 05} (2007) 039, \href{http://www.arXiv.org/abs/hep-th/0702141}{{\tt hep-th/0702141}}.

\bibitem{Banerjee:2008yu}
S.~Banerjee, A.~Sen, and Y.~K. Srivastava, ``{Genus Two Surface and Quarter BPS Dyons: The Contour Prescription},'' {\em JHEP} {\bf 03} (2009) 151, \href{http://www.arXiv.org/abs/0808.1746}{{\tt 0808.1746}}.

\bibitem{Ferrara:1995ih}
S.~Ferrara, R.~Kallosh, and A.~Strominger, ``{N=2 extremal black holes},'' {\em Phys. Rev. D} {\bf 52} (1995) R5412--R5416, \href{http://www.arXiv.org/abs/hep-th/9508072}{{\tt hep-th/9508072}}.

\bibitem{Moore:1998pn}
G.~W. Moore, ``{Arithmetic and attractors},'' \href{http://www.arXiv.org/abs/hep-th/9807087}{{\tt hep-th/9807087}}.

\bibitem{zwegers2008mock}
S.~Zwegers, ``Mock theta functions,'' 2008.

\bibitem{Sen:2011mh}
A.~Sen, ``{Negative discriminant states in N=4 supersymmetric string theories},'' {\em JHEP} {\bf 10} (2011) 073, \href{http://www.arXiv.org/abs/1104.1498}{{\tt 1104.1498}}.

\bibitem{Bhand:2023rhm}
A.~Bhand, A.~Sen, and R.~K. Singh, ``{Mock Modularity In CHL Models},'' \href{http://www.arXiv.org/abs/2311.16252}{{\tt 2311.16252}}.

\bibitem{Zagier1989}
D.~Zagier and N.-P. Skoruppa, ``A trace formula for jacobi forms.,'' {\em Journal für die reine und angewandte Mathematik} {\bf 393} (1989) 168--198.

\bibitem{rademacher}
H.~Rademacher, ``On the partition function p(n),'' {\em Proceedings of the London Mathematical Society} {\bf s2-43} (1938), no.~1, 241--254, \href{http://www.arXiv.org/abs/https://londmathsoc.onlinelibrary.wiley.com/doi/pdf/10.1112/plms/s2-43.4.241}{{\tt https://londmathsoc.onlinelibrary.wiley.com/doi/pdf/10.1112/plms/s2-43.4.241}}.

\bibitem{apostol2012modular}
T.~Apostol, {\em Modular Functions and Dirichlet Series in Number Theory}.
\newblock Graduate Texts in Mathematics. Springer New York, 2012.

\bibitem{LopesCardoso:2004law}
G.~L. Cardoso, B.~de~Wit, J.~Kappeli, and T.~Mohaupt, ``{Asymptotic degeneracy of dyonic N = 4 string states and black hole entropy},'' {\em JHEP} {\bf 12} (2004) 075, \href{http://www.arXiv.org/abs/hep-th/0412287}{{\tt hep-th/0412287}}.

\bibitem{7afd335e-bde4-3897-9875-1f8b53717c0d}
J.~H. Conway and F.~Y.~C. Fung, {\em The Sensual Quadratic Form}, vol.~26.
\newblock Mathematical Association of America, 1~ed., 1997.

\bibitem{Cvetic:1995bj}
M.~Cvetic and A.~A. Tseytlin, ``{Solitonic strings and BPS saturated dyonic black holes},'' {\em Phys. Rev. D} {\bf 53} (1996) 5619--5633, \href{http://www.arXiv.org/abs/hep-th/9512031}{{\tt hep-th/9512031}}. [Erratum: Phys.Rev.D 55, 3907 (1997)].

\bibitem{Cvetic:1995uj}
M.~Cvetic and D.~Youm, ``{Dyonic BPS saturated black holes of heterotic string on a six torus},'' {\em Phys. Rev. D} {\bf 53} (1996) 584--588, \href{http://www.arXiv.org/abs/hep-th/9507090}{{\tt hep-th/9507090}}.

\bibitem{baricz_2010}
A.~Baricz, ``Bounds for modified bessel functions of the first and second kinds,'' {\em Proceedings of the Edinburgh Mathematical Society} {\bf 53} (2010), no.~3, 575–599.

\bibitem{Bringmann:2015gla}
K.~Bringmann, L.~Rolen, and S.~Zwegers, ``{On the Fourier coefficients of negative index meromorphic Jacobi forms},'' \href{http://www.arXiv.org/abs/1501.04476}{{\tt 1501.04476}}.

\end{thebibliography}\endgroup
\bibliographystyle{utphys}

\end{document}